\documentclass[twoside,11pt,final]{entics}
\usepackage{enticsmacro}
\usepackage{graphicx}
\usepackage{tikz-cd}
\usepackage[all]{xy}
\def\conf{MFPS 2026} 	%
\volume{6}			%
\def\volu{6}			%
\def\papno{5}			%
\def\mydoi{19580}
\def\lastname{Earnshaw, Nester, Román} %
\def\runtitle{Monoidal categories graded by partial commutative monoids} %
\def\copynames{M. Earnshaw, C. Nester, M. Román}    %

\usepackage[inline]{enumitem}
\usepackage{mathrsfs}
\usepackage{amsmath}
\usepackage{stmaryrd}
\usepackage{cleveref}
\usepackage{mathtools}
\usepackage{thmtools,thm-restate}

\usepackage[xcolor,no patch,hyperref,quotation,paper]{knowledge}
\knowledgeconfigure{quotation,protect quotation={tikzcd}}
\knowledgeconfigure{notion}
\knowledgestyle{notion}{color=black}

\knowledge{notion}
| extension order

\knowledge{notion}
| PCM
| PCMs
| \PCM
| partial commutative monoid
| Partial commutative monoid
| partial commutative monoids
| Partial commutative monoids

\knowledge{notion}
| partial commutative monoid morphism
| morphism of partial commutative monoids
| homomorphism of partial commutative monoids

\knowledge{notion}
| orthogonal

\knowledge{notion}
| effectful category
| effectful categories

\knowledge{notion}
| \Eff

\knowledge{notion}
| \Freyd

\knowledge{notion}
| \SymEff

\knowledge{notion}
| effectful functor
| effectful functors

\knowledge{notion}
| symmetric strict effectful category
| symmetric strict effectful categories
| symmetric effectful category
| symmetric effectful categories

\knowledge{notion}
| symmetric strict effectful functor
| symmetric strict effectful functors

\knowledge{notion}
| strict effectful category
| strict effectful categories
| Strict effectful categories

\knowledge{notion}
| monoidal category graded by
| graded monoidal category
| -graded monoidal category
| \pcm{}-graded monoidal categories
| $E$-graded monoidal category
| $E$-graded monoidal categories
| $(E, \oplus , 0)$-graded monoidal category

\knowledge{notion}
| strict monoidal functor

\knowledge{notion}
| Freyd category
| Freyd categories

\knowledge{notion}
| convolution

\knowledge{notion}
| profunctor

\knowledge{notion}
| symmetric promonoidal category

\knowledge{notion}
| symmetric promonoidal functor

\knowledge{notion}
| promonoidal category
| promonoidal structure

\knowledge{notion}
| strict premonoidal functor

\knowledge{notion}
| premonoidal categories
| strict premonoidal categories
| strict premonoidal category
| strict premonoidal structure
| premonoidal structure
| premonoidal category

\knowledge{notion}
| monoidal product
| tensor product
| strict monoidal category
| strict monoidal structure
| monoidal category
| monoidal categories
| monoidal structure

\knowledge{notion}
| cartesian monoidal functor

\knowledge{notion}
| cartesian monoidal
| cartesian monoidal category

\knowledge{notion}
| symmetric strict monoidal category
| symmetric strict monoidal structure

\knowledge{notion}
| symmetric strict monoidal functor

\knowledge{notion}
| \SymGradMonTop

\knowledge{notion}
| symmetric strict premonoidal category
| symmetric premonoidal category
| symmetric premonoidal
| symmetric strict premonoidal structure

\knowledge{notion}
| symmetric premonoidal functor
| symmetric strict premonoidal functor

\knowledge{notion}
| \bf{1}
| \mathbf{1}

\knowledge{notion}
| \bf{2}
| \mathbf{2}
| $\bf{2}$

\knowledge{notion}
| \bf{3}
| \mathbf{3}

\knowledge{notion}
| has joins

\knowledge{notion}
| directed

\knowledge{notion}
| top element
| with top
| with top element
| has a top element

\knowledge{notion}
| \GradMonTop

\knowledge{notion}
| lax promonoidal
| lax promonoidal functor

\knowledge{notion}
| companion

\knowledge{notion}
| exclusive monoid

\knowledge{notion}
| effect algebra
| effect algebras

\knowledge{notion}
| separation algebra
| separation algebras

\knowledge{notion}
| center
| central
| Central
| central morphisms

\knowledge{notion}
| interchange
| interchanging

\knowledge{notion}
| center of a premonoidal category

\knowledge{notion}
| functor of Freyd categories

\usepackage{etoolbox}
\AtBeginEnvironment{lem}{\crefalias{thm}{lem}}
\AtBeginEnvironment{prop}{\crefalias{thm}{prop}}
\AtBeginEnvironment{cor}{\crefalias{thm}{cor}}
\AtBeginEnvironment{defn}{\crefalias{thm}{defn}}
\AtBeginEnvironment{rem}{\crefalias{thm}{rem}}
\AtBeginEnvironment{exmp}{\crefalias{thm}{exmp}}
\AtBeginEnvironment{prob}{\crefalias{thm}{prob}}

\crefname{thm}{Theorem}{Theorems}
\crefname{cor}{Corollary}{Corollaries}
\crefname{lem}{Lemma}{Lemmas}
\crefname{prop}{Proposition}{Propositions}
\crefname{defn}{Definition}{Definitions}
\crefname{rem}{Remark}{Remarks}
\crefname{exmp}{Example}{Examples}
\crefname{claim}{Claim}{Claims}
\crefname{axiom}{Axiom}{Axioms}
\crefname{conj}{Conjecture}{Conjectures}
\crefname{fact}{Fact}{Facts}
\crefname{hypo}{Hypothesis}{Hypotheses}
\crefname{assum}{Assumption}{Assumptions}
\crefname{crit}{Criterion}{Criteria}
\crefname{prob}{Problem}{Problems}
\crefname{prin}{Principle}{Principles}
\crefname{alg}{Algorithm}{Algorithms} 
\crefname{algorithm}{Algorithm}{Algorithms}
\crefname{note}{Note}{Notes}
\crefname{summ}{Summary}{Summaries}
\crefname{case}{Case}{Cases}
\crefname{equation}{Equation}{Equations}
\crefname{figure}{Figure}{Figures}
\crefname{table}{Table}{Tables}
\crefname{construction}{Construction}{Constructions}
\crefname{observation}{Observation}{Observation}

\usepackage{cmll}

\newcommand\scalemath[3]{\scalebox{#1}[#2]{\mbox{\ensuremath{\displaystyle #3}}}}

\newcommand{\leftarrowtip}{\ensuremath{\tikz\draw[line width=0.5pt,->] (10pt,0) -- (0,0);}}
\newcommand{\leftarrowtailnotip}{\ensuremath{\tikz\draw[line width=0.5pt,-<] (0,0) -- (10pt,0);}}

\newcommand{\unicodeStar}{\ensuremath{\star}}
\DeclareUnicodeCharacter{2605}{\unicodeStar}

\DeclareUnicodeCharacter{21D2}{\ensuremath{\Rightarrow}}
\DeclareUnicodeCharacter{2218}{\ensuremath{\circ}}

\DeclareUnicodeCharacter{2022}{\ensuremath{\bullet}}
\DeclareUnicodeCharacter{2219}{\ensuremath{\bullet}}
\DeclareUnicodeCharacter{2026}{\ensuremath{\dots}}
\DeclareUnicodeCharacter{2208}{\ensuremath{\in}}
\DeclareUnicodeCharacter{2192}{\ensuremath{\to}}
\DeclareUnicodeCharacter{2190}{\ensuremath{\leftarrowtip}}
\DeclareUnicodeCharacter{2919}{\ensuremath{\leftarrowtailnotip}}

\DeclareUnicodeCharacter{00D7}{\ensuremath{\times}}
\DeclareUnicodeCharacter{00B7}{\ensuremath{\cdot}}
\DeclareUnicodeCharacter{222B}{\ensuremath{\int}}
\DeclareUnicodeCharacter{22A4}{\ensuremath{\top}}
\DeclareUnicodeCharacter{22A5}{\ensuremath{\bot}}
\DeclareUnicodeCharacter{2264}{\ensuremath{\leq}}

\newcommand{\unicodecolon}{\ensuremath{\colon}}
\DeclareUnicodeCharacter{FE55}{\unicodecolon}
\newcommand{\unicodeleftpar}{\ensuremath{\left(}}
\DeclareUnicodeCharacter{27EE}{\unicodeleftpar}
\newcommand{\unicoderightpar}{\ensuremath{\right)}}
\DeclareUnicodeCharacter{27EF}{\unicoderightpar}
\DeclareUnicodeCharacter{2260}{\neq}
\DeclareUnicodeCharacter{22A9}{\Vdash}
\DeclareUnicodeCharacter{2237}{\proportion}
\DeclareUnicodeCharacter{2124}{\mathbb{Z}}
\DeclareUnicodeCharacter{27E8}{\langle}
\DeclareUnicodeCharacter{27E9}{\rangle}
\DeclareUnicodeCharacter{21A6}{\mapsto}
\DeclareUnicodeCharacter{22A2}{\vdash}
\DeclareUnicodeCharacter{2090}{\ensuremath{{}_a}}
\DeclareUnicodeCharacter{A71B}{{}^\uparrow}
\DeclareUnicodeCharacter{A71C}{{}^\downarrow}
\DeclareUnicodeCharacter{27E6}{\llbracket}
\DeclareUnicodeCharacter{27E7}{\rrbracket}

\newcommand{\unicoderightcircle}{\ensuremath{\RIGHTcircle}}
\DeclareUnicodeCharacter{25D1}{\unicoderightcircle}
\newcommand{\unicodeleftcircle}{\ensuremath{\LEFTcircle}}
\DeclareUnicodeCharacter{25D0}{\unicodeleftcircle}
\DeclareUnicodeCharacter{229B}{\circledast}

\newcommand{\unicodebbA}{\ensuremath{\mathbb{A}}}
\DeclareUnicodeCharacter{1D538}{\unicodebbA}
\newcommand{\unicodebbB}{\ensuremath{\mathbb{B}}}
\DeclareUnicodeCharacter{1D539}{\unicodebbB}
\newcommand{\unicodebbC}{\ensuremath{\mathbb{C}}}
\DeclareUnicodeCharacter{2102}{\unicodebbC}
\DeclareUnicodeCharacter{1D53B}{\ensuremath{\mathbb{D}}}
\DeclareUnicodeCharacter{2115}{\ensuremath{\mathbb{N}}}
\DeclareUnicodeCharacter{211D}{\ensuremath{\mathbb{R}}}
\DeclareUnicodeCharacter{1D543}{\ensuremath{\mathbb{L}}}
\newcommand\UnicodeBlackboardP{\ensuremath{\mathbf{P}}} \DeclareUnicodeCharacter{2119}{\UnicodeBlackboardP}
\DeclareUnicodeCharacter{211A}{\ensuremath{\mathbb{Q}}}
\DeclareUnicodeCharacter{1D544}{\ensuremath{\mathbb{M}}}
\DeclareUnicodeCharacter{1D54C}{\ensuremath{\mathbb{U}}}
\DeclareUnicodeCharacter{1D54D}{\ensuremath{\mathbf{V}}}
\DeclareUnicodeCharacter{1D54E}{\ensuremath{\mathbb{W}}}
\DeclareUnicodeCharacter{1D546}{\ensuremath{\mathbb{O}}}
\DeclareUnicodeCharacter{1D540}{\ensuremath{\mathbb{I}}}
\DeclareUnicodeCharacter{1D54A}{\ensuremath{\mathbb{S}}}
\DeclareUnicodeCharacter{1D53C}{\ensuremath{\mathbb{E}}}

\newcommand{\unicodecalS}{\ensuremath{\mathcal{S}}}
\newcommand{\unicodecalT}{\ensuremath{\mathcal{T}}}
\newcommand{\unicodecalC}{\ensuremath{\mathcal{C}}}
\newcommand{\unicodecalD}{\ensuremath{\mathcal{D}}}
\newcommand{\unicodecalX}{\ensuremath{\mathcal{X}}}
\newcommand{\unicodecalN}{\ensuremath{\mathcal{N}}}
\newcommand{\unicodecalE}{\ensuremath{\mathcal{E}}}
\DeclareUnicodeCharacter{1D4D4}{\unicodecalE}
\DeclareUnicodeCharacter{1D4D2}{\unicodecalC}
\DeclareUnicodeCharacter{1D4D3}{\unicodecalD}
\DeclareUnicodeCharacter{1D4DE}{\mathcal{O}}
\DeclareUnicodeCharacter{1D4AA}{\mathcal{O}}
\DeclareUnicodeCharacter{210B}{\mathcal{H}}
\DeclareUnicodeCharacter{1D4E2}{\unicodecalS}
\DeclareUnicodeCharacter{1D4E3}{\unicodecalT}
\DeclareUnicodeCharacter{1D4E7}{\unicodecalX}
\DeclareUnicodeCharacter{1D4D0}{\ensuremath{\mathcal{A}}}
\DeclareUnicodeCharacter{1D4D1}{\ensuremath{\mathcal{B}}}
\DeclareUnicodeCharacter{1D4D6}{\ensuremath{\mathcal{G}}}
\DeclareUnicodeCharacter{1D4D7}{\ensuremath{\mathcal{H}}}
\DeclareUnicodeCharacter{1D4DB}{\ensuremath{\mathcal{L}}}
\DeclareUnicodeCharacter{1D4DC}{\ensuremath{\mathcal{M}}}
\DeclareUnicodeCharacter{1D4DD}{\unicodecalN}
\DeclareUnicodeCharacter{1D4B1}{\ensuremath{\mathcal{V}}}
\DeclareUnicodeCharacter{1D4E5}{\ensuremath{\mathcal{V}}}
\DeclareUnicodeCharacter{1D4E6}{\ensuremath{\mathcal{W}}}
\DeclareUnicodeCharacter{1D4E4}{\ensuremath{\mathcal{U}}}

\DeclareUnicodeCharacter{03B1}{\alpha}
\DeclareUnicodeCharacter{03B2}{\beta}
\DeclareUnicodeCharacter{03BC}{\mu}
\DeclareUnicodeCharacter{03B4}{\delta}
\DeclareUnicodeCharacter{03B5}{\varepsilon}
\DeclareUnicodeCharacter{03B7}{\eta}
\DeclareUnicodeCharacter{03BB}{\lambda}
\DeclareUnicodeCharacter{03C1}{\rho}
\DeclareUnicodeCharacter{03C8}{\psi}
\DeclareUnicodeCharacter{03C4}{\tau}
\DeclareUnicodeCharacter{03A8}{\Psi}
\DeclareUnicodeCharacter{03C3}{\sigma}
\DeclareUnicodeCharacter{03C6}{\varphi}
\DeclareUnicodeCharacter{03A6}{\Phi}
\DeclareUnicodeCharacter{03A3}{\Sigma}
\DeclareUnicodeCharacter{03D5}{\phi}
\DeclareUnicodeCharacter{03B8}{\theta}
\DeclareUnicodeCharacter{03C0}{\ensuremath{\pi}}
\DeclareUnicodeCharacter{0393}{\Gamma}
\DeclareUnicodeCharacter{0394}{\Delta}
\DeclareUnicodeCharacter{03BA}{\kappa}
\DeclareUnicodeCharacter{03BD}{\nu}
\DeclareUnicodeCharacter{03A9}{\Omega}
\DeclareUnicodeCharacter{03C9}{\omega}
\DeclareUnicodeCharacter{03D2}{\Upsilon}
\DeclareUnicodeCharacter{03C5}{\upsilon}

\DeclareUnicodeCharacter{25A0}{\blacksquare}
\DeclareUnicodeCharacter{25AA}{\blacksquare}
\DeclareUnicodeCharacter{2205}{\varnothing}
\DeclareUnicodeCharacter{2254}{\coloneq}
\DeclareUnicodeCharacter{03B3}{\gamma}

\newcommand{\hirayo}{\scaleobj{0.9}{\text{\usefont{U}{min}{m}{n}\symbol{'210}}}}
\DeclareUnicodeCharacter{3088}{\hirayo}
\DeclareFontFamily{U}{min}{}
\DeclareFontShape{U}{min}{m}{n}{<-> udmj30}{}

\newcommand\UnicodeWhiteRightPointingSmallTriangle{\triangleright}
\DeclareUnicodeCharacter{25B9}{\mathbin{\UnicodeWhiteRightPointingSmallTriangle}}
\newcommand\UnicodeWhiteDownPointingSmallTriangle{\triangledown}
\DeclareUnicodeCharacter{25BF}{\mathbin{\UnicodeWhiteDownPointingSmallTriangle}}
\newcommand\UnicodeWhiteUpPointingSmallTriangle{\scalemath{1}{-1}{{}^{\triangledown}}}
\DeclareUnicodeCharacter{25B5}{\mathbin{\UnicodeWhiteUpPointingSmallTriangle}}

\DeclareUnicodeCharacter{2080}{\ensuremath{{}_0}}
\DeclareUnicodeCharacter{2081}{\ensuremath{{}_1}}
\DeclareUnicodeCharacter{2082}{\ensuremath{{}_2}}
\DeclareUnicodeCharacter{2083}{\ensuremath{{}_3}}

\DeclareUnicodeCharacter{1D62}{\ensuremath{{}_i}}
\DeclareUnicodeCharacter{2C7C}{\ensuremath{{}_j}}
\DeclareUnicodeCharacter{02B3}{\ensuremath{{}^r}}
\DeclareUnicodeCharacter{02E1}{\ensuremath{{}^\ell}}
\DeclareUnicodeCharacter{1D48}{\ensuremath{{}^d}}
\DeclareUnicodeCharacter{1D50}{\ensuremath{{}^m}}
\DeclareUnicodeCharacter{1D58}{\ensuremath{{}^u}}
\DeclareUnicodeCharacter{209A}{\ensuremath{{}_p}}
\DeclareUnicodeCharacter{2096}{\ensuremath{{}_k}}

\DeclareUnicodeCharacter{2245}{\ensuremath{\cong}}
\DeclareUnicodeCharacter{2286}{\subseteq}

\DeclareUnicodeCharacter{22C5}{\cdot}
\DeclareUnicodeCharacter{25C3}{\ensuremath{\triangleleft}}
\DeclareUnicodeCharacter{25B9}{\ensuremath{\triangleright}}

\DeclareUnicodeCharacter{2217}{\ast}
\DeclareUnicodeCharacter{039E}{\Xi}
\DeclareUnicodeCharacter{2295}{\oplus}
\DeclareUnicodeCharacter{2297}{\otimes}
\DeclareUnicodeCharacter{214B}{\parr}
\DeclareUnicodeCharacter{2298}{\oslash}
\DeclareUnicodeCharacter{25C0}{\mathbin{\blacktriangleleft}}
\DeclareUnicodeCharacter{25C1}{\mathbin{\vartriangleleft}}
\DeclareUnicodeCharacter{22B3}{\mathbin{\triangleright}}
\DeclareUnicodeCharacter{22B2}{\mathbin{\triangleleft}}
\DeclareUnicodeCharacter{FF5C}{\mid}
\DeclareUnicodeCharacter{227A}{\mathbin{\prec}}
\DeclareUnicodeCharacter{227B}{\mathbin{\succ}}
\DeclareUnicodeCharacter{22A3}{\mathbin{\dashv}}
\DeclareUnicodeCharacter{219D}{\ensuremath{\leadsto}}
\DeclareUnicodeCharacter{2191}{\ensuremath{\uparrow}}
\DeclareUnicodeCharacter{1361}{\colon}

\DeclareUnicodeCharacter{29D1}{\mathrel{\multimapdotbothB}}
\DeclareUnicodeCharacter{29D2}{\mathrel{\multimapdotbothA}}
\DeclareUnicodeCharacter{22C4}{\mathbin{\diamond}}
\DeclareUnicodeCharacter{226B}{\mathrel{\gg}}
\DeclareUnicodeCharacter{25A1}{\mathbin{\square}}
\DeclareUnicodeCharacter{266F}{\sharp}

\DeclareUnicodeCharacter{2113}{\ell}
\DeclareUnicodeCharacter{2261}{\equiv}
\DeclareUnicodeCharacter{2099}{_n}
\DeclareUnicodeCharacter{2098}{_m}
\DeclareUnicodeCharacter{1D5AD}{\ensuremath{\mathsf{N}}}

\DeclareUnicodeCharacter{1D4DF}{\mathcal{P}}
\DeclareUnicodeCharacter{1D415}{\mathbf{V}}
\DeclareUnicodeCharacter{1D400}{\mathbf{A}}
\DeclareUnicodeCharacter{1D6C2}{\pmb{\alpha}}
\DeclareUnicodeCharacter{1D6C3}{\pmb{\beta}}
\DeclareUnicodeCharacter{1D6C4}{\pmb{\gamma}}
\DeclareUnicodeCharacter{1D6C5}{\pmb{\delta}}
\DeclareUnicodeCharacter{1D6C6}{\pmb{\epsilon}}
\DeclareUnicodeCharacter{1D6C8}{\pmb{\eta}}
\DeclareUnicodeCharacter{1D6DA}{\pmb{\omega}}

\newcommand\mydots{\makebox[0.6em][c]{.\hfil.\hfil.}}
\DeclareUnicodeCharacter{2026}{\mydots}
\DeclareUnicodeCharacter{226B}{\gg}

\DeclareUnicodeCharacter{22C9}{\ltimes}
\DeclareUnicodeCharacter{22CA}{\rtimes}

\newcommand{\unicodeRelationalComposition}{\fatsemi}
\DeclareUnicodeCharacter{2A3E}{\unicodeRelationalComposition}

\newcommand{\Pcms}{\kl{PCMs}}
\newcommand{\pcm}{\kl{PCM}}
\newcommand{\pcms}{\kl{PCMs}}

\newcommand{\Bool}{\mathsf{Bool}}
\newcommand{\join}{\vee}
\newcommand{\comp}{\fatsemi}
\newcommand{\obj}{_{\text{obj}}}
\newcommand{\id}{\mathrm{id}}
\newcommand{\bid}{\mathrm{bid}}
\newcommand{\undefi}{\uparrow}
\newcommand{\regaux}[2]{\smash{\begin{array}{@{}l@{}}
  \scriptstyle #2\\[-1.04em] %
  \scriptstyle #1
\end{array}}} %

\newcommand{\reg}[2]{\raisebox{.12em}{$\regaux{#1}{#2}$}}

\newcommand{\cat}[1]{\mathsf{#1}}

\newcommand{\GradMon}{\cat{GradMon}}

\newcommand{\MonCatl}{\cat{MonCat}_{\text{lax}}}

\newcommand{\TwoGradMon}{\mathbf{2}\text{-}\cat{GradMon}}
\newcommand{\TwoSymGradMon}{\cat{Sym}\mathbf{2}\text{-}\cat{GradMon}}
\newcommand{\TwoCartGradMon}{\cat{Cart}\mathbf{2}\text{-}\cat{GradMon}}
\newcommand{\GradMonTop}{\cat{GradMon}_{\top}}
\newcommand{\SymGradMonTop}{\cat{SymGradMon}_{\top}}
\newcommand{\CartGradMonTop}{\cat{CartGradMon}_{\top}}

\newcommand{\PCM}{\cat{PCM}}

\newcommand{\Cat}{\cat{Cat}}
\newcommand{\Set}{\cat{Set}}
\newcommand{\Eff}{\cat{Eff}}
\newcommand{\SymEff}{\cat{SymEff}}
\newcommand{\Freyd}{\cat{Freyd}}

\newcommand{\op}{^\mathsf{op}}

\newcommand{\nperp}{\mathbin{\rotatebox{90}{\ensuremath{\nvdash}}}}

\usepackage{tikz}

\newcommand{\venturieq}{
  \mathrel{
    \begin{tikzpicture}[baseline=-0.25em, line width=0.06em]
      \pgfmathsetmacro{\gap}{0.008ex}
      \pgfmathsetmacro{\len}{0.015em}
      \pgfmathsetmacro{\tip}{0.01em}
      \pgfmathsetmacro{\spread}{0.008em}
      \draw (-\tip, \spread) -- (0, \gap) -- (\len, \gap);
      \draw (-\tip, -\spread) -- (0, -\gap) -- (\len, -\gap);
    \end{tikzpicture}
  }
}

\usepackage{trimclip}
\makeatletter
\DeclareRobustCommand{\shortto}{%
  \mathpalette\short@to\relax
}

\newcommand{\short@to}[2]{%
  \clipbox{{0.4\width} 0 0 0}{$\m@th#1\vphantom{+}{\rightarrow}$}%
  }
\makeatother

 \makeatletter
 \def\slashedarrowfill@#1#2#3#4#5{%
 $\m@th\thickmuskip0mu\medmuskip\thickmuskip\thinmuskip\thickmuskip
 \relax#5#1\mkern-7mu%
  \cleaders\hbox{$#5\mkern-2mu#2\mkern-2mu$}\hfill
  \mathclap{#3}\mathclap{#2}%
 \cleaders\hbox{$#5\mkern-2mu#2\mkern-2mu$}\hfill
    \mkern-7mu#4$%
    }
    \def\rightslashedarrowfill@{%
    \slashedarrowfill@\relbar\relbar\mapstochar\rightarrow}
    \newcommand\xslashedrightarrow[2][]{%
    \ext@arrow 0055{\rightslashedarrowfill@}{#1}{#2}}
    \def\leftslashedarrowfill@{%
    \slashedarrowfill@\leftarrow\relbar\mapsfromchar\relbar}
    \newcommand\xslashedleftarrow[2][]{%
    \ext@arrow 0055{\leftslashedarrowfill@}{#1}{#2}}
    \makeatother

    \newcommand{\xlto}{\xslashedrightarrow}
    \newcommand{\lto}{\xlto{}}

\begin{document}

\begin{frontmatter}
  \title{Monoidal Categories Graded by \\ Partial Commutative Monoids}%
  \author{Matthew Earnshaw\thanksref{a}\thanksref{ta}}%
  \author{Chad Nester\thanksref{a}\thanksref{ta}}
  \author{Mario Román\thanksref{b}\thanksref{c}\thanksref{tb}}
   \address[a]{Institute of Computer Science, University of Tartu, Estonia}  							
  \address[b]{Department of Software Science, Tallinn University of Technology, Estonia} 
  \address[c]{Department of Computer Science, University of Oxford, United Kingdom}
  \thanks[ta]{Supported by Estonian Research Council grant PRG2764.}
  \thanks[tb]{Supported by the Safeguarded AI programme of the Advanced Research + Invention Agency, and the Estonian Research Council grant PRG3215.}
\begin{abstract} 
  Effectful categories have two classes of morphisms: \emph{pure} morphisms, which form a monoidal category; and \emph{effectful} morphisms, which can only be combined monoidally with central morphisms (such as the pure ones), forming a premonoidal category. This suggests seeing morphisms of an effectful category as carrying a \emph{grade} that combines under the monoidal product in a \emph{partially defined} manner. We axiomatize this idea with the notion of \emph{monoidal category graded by a partial commutative monoid} (PCM). Strict monoidal categories arise as the special case of grading by the singleton PCM, and strict effectful categories arise from grading by a two-element PCM. Further examples include grading by powerset PCMs, modelling non-interfering parallelism for programs accessing shared resources, and grading by intervals, modelling bounded resource usage.
  We show that effectful categories form a coreflective subcategory of PCM-graded monoidal categories; introduce cartesian structure, recovering Freyd categories; and describe PCM-graded monoidal categories as monoids by viewing a PCM as a thin promonoidal category.%
\end{abstract}
\end{frontmatter}

\maketitle

\section{Introduction}
\emph{Effectful categories}, or \emph{generalized Freyd categories}, refine monoidal categories into a structure fit for the semantics of effectful programming languages \cite{jeffrey1997premonoidal1,LEVY2003182,roman2022promonads}. They do this by dividing morphisms between a category of \emph{pure computations} and a category of \emph{effectful computations}, with a functor including the former amongst the latter. Since pure computations are independent of one another, their parallel execution is well defined, and is modelled by a monoidal product. Effectful computations instead form a \emph{premonoidal category}: since in general they depend on one another, these morphisms may only be combined monoidally with \emph{central} morphisms (those \emph{interchanging} with all other morphisms), such as morphisms in the image of the functor from the category of pure computations \cite{power_robinson_1997}. This
two-part %
categorization of morphisms
suggests taking the perspective that morphisms in an effectful category have an algebraic \emph{grading}.

In the existing literature on \emph{graded effect systems} \cite{kammar2012algebraic,10.1145/73560.73564,mcdermott:LIPIcs.FSCD.2025.28,10.1145/601775.601776} and their semantics in \emph{graded monads} \cite{katsumata2014parametric,mellies2017,orchard2014semantic,Smirnov2008}, grades typically combine under \emph{sequential composition} of morphisms. The grading introduced here is different in two respects. Firstly, grades combine only under the \emph{monoidal product} of morphisms. The monoidal product of a pure and an effectful computation, for example, should yield an effectful computation. On the other hand, sequentially composing two pure computations should yield a pure computation, and likewise for effectful computations. Secondly, the combination of grades need only be \emph{partially defined}: for example, the monoidal product of two effectful computations must be undefined. The central contribution of this paper, laid out starting in \Cref{sec:gradpcm}, is the notion of \emph{monoidal category graded by a partial commutative monoid (PCM)}, which axiomatizes this idea.

In a monoidal category graded by a PCM, every morphism has a grade, taken from a partial commutative monoid $(E, \oplus, 0)$. Thus they comprise a family of categories $\{\mathbb{C}_a\}_{a \in E}$ indexed by the grading PCM, and monoidal product operations of the type
$$(\otimes)_{a,b} : \mathbb{C}_a(X;Y) \times \mathbb{C}_b(X';Y') \to \mathbb{C}_{a \oplus b}(X \otimes X' ; Y \otimes Y').$$
Crucially, $(\otimes)_{a,b}$ exists only when $a \oplus b$ is defined in the grading PCM. If we grade by the singleton PCM, $\mathbf{1}$, we recover strict monoidal categories. Strict effectful categories are isomorphic to strict monoidal categories graded by the powerset PCM, $\bf{2} \cong \mathscr{P}(1)$, in which $1 \oplus 1$ is undefined.

More generally, we can consider the powerset \pcm{}, $\mathscr{P}(X)$, over an arbitrary set $X$, whose operation of union is defined only on disjoint subsets. A monoidal category graded by the powerset \pcm{} models safe parallelism of programs accessing a set of heap locations, file handles or other such \emph{devices} in $X$: their monoidal product exists only when they use disjoint devices. Grading by an interval $[0,r]$, which forms a \pcm{} under \emph{bounded addition}, captures the situation of bounded \emph{bandwidth}, where morphisms are graded by their bandwidth usage, and may use no more than the bound \emph{in parallel}. All of these examples, and more, are given in more detail in \Cref{sec:gradpcm}.

In \Cref{sec:effectful}, we establish an isomorphism between the category of effectful categories and the category of "$\bf{2}$"-graded monoidal categories, and exhibit the latter as a coreflective subcategory of \pcm{}-graded monoidal categories. We go on to introduce symmetric and cartesian structure on \pcm{}-graded monoidal categories, extending the case of grading by $\bf{2}$ to find an isomorphism between the category of cartesian $"\bf{2}"$-graded monoidal categories and the category of Freyd categories.

In \Cref{sec:globalcat} we show that for certain well behaved \pcms{} $E$, we can assemble the family of categories $\{\mathbb{C}_a\}_{a \in E}$ of an $E$-graded monoidal category into a single category in two ways, which makes sense of the sequential composition of heterogeneously graded morphisms.

\Cref{sec:gradprom} lays out a more formal perspective on our central definition, reformulating \pcm{}-graded monoidal categories as monoids in a monoidal category. The key idea in this reformulation is to view a \pcm{} as a thin promonoidal category. This is enough structure to obtain a monoidal structure on presheaves, which forms the heart of this reformulation.

\subsection{Related work}

Premonoidal categories were introduced by Power and Robinson \cite{power_robinson_1997} as a reformulation of Moggi’s monadic semantics of effectful programming languages \cite{moggi1991notions}, capturing structure present in the Kleisli category of strong monads. Freyd categories, introduced by Levy, Power and Thielecke \cite{10.1007/BFb0014560,levy2004}, extend these with a cartesian base and inclusion functor, providing semantics for effectful call-by-value languages. Jeffrey \cite{jeffrey1997premonoidal1} considered an additional symmetric monoidal category of central computations to model control-flow graphs. Bonchi, %
Di Lavore and Román \cite{11186351} have used the resulting \emph{effectful triples} to define effectful Mealy machines. Non-cartesian Freyd categories or \emph{effectful categories} have recently %
been applied to the semantics of SSA by Ghalayini and Krishnaswami \cite{ghalayini2024denotationalsemanticsssa}. 

A substantial literature studies the syntax and semantics of languages in which programs are equipped with grades tracking quantitative or qualitative information such as effects, costs, or security levels. Key examples include graded monads (Katsumata \cite{katsumata2014parametric}; Melliès \cite{mellies2017}; Orchard, Petricek and Mycroft \cite{orchard2014semantic}), indexed monads (Maillard and Melliès \cite{10.1109/LICS.2015.45}), parameterized monads (Atkey \cite{ATKEY_2009}), and graded Freyd categories \cite{gaboardi2021graded}, where grades combine under sequential composition. By contrast, sequential composition in PCM-graded monoidal categories preserves the grade. Grades instead combine under the monoidal product, and moreover govern its existence: the monoidal product is defined only when the sum of the grades is defined in the grading \pcm{}. This makes \pcm{}-graded monoidal categories well suited to modelling ambient resources such as memory locations or bandwidth, rather than sequentially accumulating information.

\Pcms{}, especially qua \emph{separation} and \emph{effect algebras}, are widely used to model shared resources or \emph{ghost state} in program logics and verification \cite{dinsdale2013views,10.1145/2398856.2364536,10.1007/978-3-642-28869-2_19,10.1145/2429069.2429134}. Undefinedness typically arises when attempting to combine sets of resources that are not disjoint, as in Reynolds' separation logic \cite{10.5555/645683.664578,calcagno2007local,PYM2004257,10.1145/2676726.2676980}, with heaps providing the paradigmatic example.

Partial monoidal categories have been defined in the literature on categorical quantum mechanics (Coecke and Lau \cite{Coecke2013-COECCR}). These restrict the monoidal product to a full subcategory of the product category, modelling space-like separation of resources. Hefford and Kissinger \cite{Hefford:2022wan} relate these to promonoidal categories, showing they generally differ but agree in special cases. In our notion, objects form a monoid and so always have a monoidal product, and the monoidal product of morphisms is controlled by the grading.

Sarkis and Zanasi \cite{sarkis_et_al:LIPIcs.CALCO.2025.5,sarkis2026gradedstringdiagramsimprecise} study monoidal categories graded by a symmetric strict monoidal category. Unlike in our notion, these grades combine totally, and by the same operation, under both sequential and monoidal composition. %

In prior work \cite{earnshaw2025resourceful}, we showed that every effectful category has an underlying signature which is effectively ``graded by'' a powerset, with a left adjoint that constructs free effectful categories. This relies on a weak notion of signature morphism; its connection to the present work remains to be clarified further. %

Heunen and Sigal \cite{heunen2023duoidally} defined enriched Freyd categories over a duoidal category $\mathcal{V}$. In \Cref{sec:gradprom}, we show that PCM-graded monoidal categories can be captured as an instance of this definition, where $\mathcal{V}$ is the category of presheaves on a promonoidal category encoding the PCM.

\section{Partial commutative monoids} \label{sec:pcms}
This section introduces \pcms{} and some basic propositions on them. Readers may wish to start at \Cref{sec:gradpcm}, consulting this section as a reference.

To avoid proliferation of side-conditions on definedness when working with \pcms{}, it is helpful to introduce the following ``Kleene equality'' notation.

\begin{definition}
  For partial functions $f, g : A \rightharpoonup B$, we write $f(a) \simeq g(a)$ to denote that if either side is defined then both are, and they are equal.
  We write $f(a) \venturieq g(a)$ to denote that if the left-hand side is defined, then so is the right, and they are equal. We write $f(a)\!\undefi$ to denote that $f$ is undefined at $a$.
\end{definition}

\begin{definition}
  \AP A \intro{partial commutative monoid} $(E, \oplus, 0)$ is a set $E$, a
  partial function $\oplus : E \times E \rightharpoonup E$ and an element $0 \in E$
  satisfying
  \begin{gather*}
  a \oplus b \simeq b \oplus a, \\
  a \oplus 0 = a = 0 \oplus a, \\
  (a \oplus b) \oplus c \simeq a \oplus (b \oplus c).
  \end{gather*}
  In view of associativity, we may unambiguously write $a \oplus b \oplus c$.
  We write $a \perp b$ (``$a$ is ""orthogonal"" to $b$'') when $a \oplus b$ is defined; conversely, we write $a \nperp b$ when $a \oplus b$ is undefined.
\end{definition}

 Of course, a commutative monoid is a \pcm{}, and we shall consider some total examples in the following.

\begin{definition}
  A ""homomorphism of partial commutative monoids"" is a (total) function $f : E \to E'$ satisfying $f(0_E) = 0_{E'}$ and $f(a \oplus b) \venturieq f(a) \oplus f(b).$
\end{definition}

\Pcms{} and their homomorphisms form a category, $"\PCM"$. Every \pcm{} gives rise to a canonical preorder, which we will use extensively in the following.

\begin{definition}
  The ""extension order"" $(E,\leqslant)$ on the elements of a partial commutative monoid $(E, \oplus, 0)$ is the preorder defined by
  $a \leqslant b \text{ if and only if there exists } c \text{ such that } a \oplus c = b.$
\end{definition}

\begin{lemma} \label{lem:leqslant-cong}
  For every element $b$ of a \pcm{} $(E, \oplus, 0)$, the operations $(- \oplus b)$ are monotonic with respect to the "extension order": $x \leqslant y$ and $y \perp b$ implies $x \perp b$ and $x \oplus b \leqslant y \oplus b$.
\end{lemma}
\begin{proof}
  Since $x \leqslant y$ we have $\exists c$, $x \oplus c = y$ by definition.
  Let $b$ be an element of $E$ and $y \perp b$. Then $y \oplus b = (x \oplus c) \oplus b = (x \oplus b) \oplus c$, applying associativity twice and commutativity.
  That is, $x \oplus b \leqslant y \oplus b$.
\end{proof}

\begin{lemma} \label{lem:oplus-monotone}
  If $a \leqslant c$, $b \leqslant d$, and $c \perp d$, then $a \perp b$ and $a \oplus b \leqslant c \oplus d$.
\end{lemma}
\begin{proof}
  By \Cref{lem:leqslant-cong}, $a \perp d$ and $a \oplus d \leqslant c \oplus d$. By commutativity and \Cref{lem:leqslant-cong} again, $a \perp b$ and $a \oplus b \leqslant a \oplus d$. We conclude by transitivity.
\end{proof}

Often we shall have $a \leqslant b$ just when $a$ is a ``less capacious'' grade than $b$, that is, $b$ ``extends'' $a$. Clearly, $0$ is a least element in the "extension order". \Pcms{} with a top element in their "extension order", such as \emph{effect algebras}, feature prominently in the following. Let us now introduce a few examples of \pcms{}.

\begin{example} \label{pcmone}
  The \pcm{} $""\bf{1}""$ is the singleton with the unique total operation.
\end{example}

\begin{example} \label{pcm2}
  The \pcm{} $""\bf{2}""$ has two elements ($0$ and $1$) with partial operation,
  \begin{gather*}
    0 \oplus 0 = 0 \qquad 0 \oplus 1 = 1 \oplus 0 = 1 \qquad 1 \nperp 1.
  \end{gather*}
\end{example}

The \pcm{} $\bf{2}$ is isomorphic to the powerset \pcm{} of $\bf{1}$, defined as follows.

\begin{example} \label{powersetpcm}
  The powerset of a set, $\mathscr{P}(X)$, has a "partial commutative monoid" structure $(\uplus, \varnothing)$ defined by taking the union of subsets only when they are disjoint
  $$S \uplus T := \begin{cases} S \cup T & \text{if } S \cap T = \varnothing \\ \undefi & \text{otherwise}.\end{cases}$$
\end{example}

\begin{definition} \label{prodpcm}
  Given a family of \pcms{} $\{(A_i,\oplus_i,0_i)\}_{i \in I}$, their product has carrier
  $\prod_{i \in I} A_i$ and operation
  $$
    (a_i)_{i \in I} \oplus (b_i)_{i \in I} := \begin{cases}
      (a_i \oplus_i b_i)_{i \in I} & \text{if } a_i \perp b_i \text{ for all } i \in I, \\
      \undefi & \text{otherwise,}
    \end{cases} \qquad   \text{with unit}\ (0_i)_{i \in I}.
  $$
\end{definition}

\emph{Separation algebras} and \emph{effect algebras} are \pcms{} with extra properties/structure making them particularly well behaved as gradings.

\begin{definition} \textup{(Calcagno, O'Hearn, Yang \cite{calcagno2007local})}
  A \intro{separation algebra} is a \pcm{} that is cancellative: if $a \oplus c = b \oplus c$ then $a = b$.
\end{definition}

\begin{definition} \textup{(Foulis and Bennett \cite{Foulis1994})}
  An ""effect algebra"" is a \pcm{} $(E,\oplus,0)$ equipped with a unary operation $(-)^{⊥} ፡ E → E$; such that $a^\bot$ is the unique element such that $a \oplus a^\bot = 1$, where $1 := 0^{\bot}$; and if $1 \oplus a$ is defined, then $a = 0$.
\end{definition}

Powerset \pcms{} are separation algebras, for example. It follows that an element $c$ witnessing $a \leqslant b$ in the "extension order" of a "separation algebra" is \emph{unique}. If the separation algebra is moreover \emph{positive} ($a \oplus b = 0$ implies $a = b = 0$), then the "extension order" is a \emph{poset}. Every "effect algebra" is a (positive) "separation algebra" \cite[Theorem 2.5]{Foulis1994}. In the "extension order" of an "effect algebra", $1$ is the "top element", as witnessed by $x \oplus x^{\bot} = 1$. The archetypical "effect algebras" are \emph{intervals} (see \Cref{ex:intgrad}).

\section{Monoidal categories graded by partial commutative monoids} \label{sec:gradpcm}

In this section we introduce our central definition, its basic properties, and several examples. We shall refer to elements of \pcms{} as \emph{grades}. \emph{Partiality} of their operation will correspond to partial definedness of monoidal products. \emph{Commutativity} will correspond to the fact that the grade of a monoidal product does not change if the factors are swapped.

\begin{definition} \label{gradmoncat}
  \AP For a "partial commutative monoid" $(E,\oplus,0)$, an ""$E$-graded monoidal category"" consists of
  \begin{itemize}
  \item a monoid of objects, $(ℂ\obj, ⊗, I)$, 
  \item for each grade, $a ∈ E$, a category $ℂ_a$ with set of objects $ℂ\obj$, with composition denoted by $$(⨾)_a ፡ ℂ_a(X;Y) × ℂ_a(Y;Z) → ℂ_a(X;Z),$$ and identities at grade 0 denoted $\id_{X}$,
    \item for each $a \leqslant b$ in the "extension order" an identity-on-objects regrading functor $$(-)_a^{b}  : ℂ_{a} → ℂ_{b},$$ allowing us to denote identities at grade $a$ by $(\id_{X})\reg{0}{a}$, and
  \item  monoidal product operations for every pair of grades $a$ and $b$ such that $a \perp b$
  $$(⊗)_{a,b} ፡ ℂ_a(X;Y) × ℂ_b(X';Y') → ℂ_{a ⊕ b}(X ⊗ X';Y ⊗ Y').$$
  \end{itemize}
  These are subject to the following axioms (with subscripts omitted on $\otimes$ and $\comp$):
  \setlist[description]{font=\normalfont\textsc}
  \begin{description}[leftmargin=2.5cm, labelwidth=2cm, labelsep=0.5cm, itemsep=0.1cm]
  \item[(\textsc{Reg-Act})] \label{regfunc} $f\reg{a}{a} = f$ and $(f\reg{a}{b})\reg{b}{c} = f\reg{a}{c}$, for $f \in ℂ_a(X;Y)$, whenever $a \leqslant b \leqslant c$.

  \item[(\textsc{Reg-$\otimes$})] \label{regotimes} $(f \otimes g)\reg{a \oplus b}{c \oplus d} = f\reg{a}{c} \otimes g\reg{b}{d}$, for $f \in ℂ_a(X;Y), g \in ℂ_b(X';Y')$, whenever $a \leqslant c$, $b \leqslant d$, and $c \perp d$ (hence well typed by \Cref{lem:oplus-monotone}).

  \item[($\otimes$-\textsc{U-A})] \label{otimesunitassoc} $f \otimes \id_I = f = \id_I \otimes f$ and $(f \otimes g) \otimes h = f \otimes (g \otimes h)$, for $f \in ℂ_a(X;Y)$, $g \in ℂ_b(X';Y')$ and $h \in ℂ_c(X'';Y'')$, and whenever $(a \oplus b) \oplus c$, hence $a \oplus (b \oplus c)$, is defined.

  \item[($\otimes$-\textsc{ID})] \label{otimesid} $\id_X \otimes \id_Y = \id_{X \otimes Y}$.

  \item[(\textsc{Inter})] \label{inter} $(f \otimes g) \comp (h \otimes k) = (f \comp h) \otimes (g \comp k)$ for $f \in ℂ_a(X;Y)$, $h \in ℂ_a(Y;Z)$, $g \in ℂ_b(X';Y')$, and $k \in ℂ_b(Y';Z')$, whenever $a \perp b$.
\end{description}

\end{definition}

The crucial point of \Cref{gradmoncat} is that $\otimes_{a,b}$ only exists when $a \perp b$: the \pcm{} structure on grades controls our ability to take the monoidal product of morphisms.
The data of a \pcm{}-graded monoidal category can be formulated in terms of a monoid in a monoidal category of lax monoidal functors, as in \Cref{sec:gradprom}.

We can show that regradings are given by monoidal products with the identity on $I$ in some grade.

\begin{proposition}
 Let $f \in \mathbb{C}_a$ be a morphism in an "$E$-graded monoidal category", where $a \leqslant b$. Then $f\reg{a}{b} = f \otimes (\id_I)\reg{0}{c}$, for every $c$ witnessing $a \leqslant b$.
\end{proposition}
\begin{proof}
    $f \otimes (\id_I)\reg{0}{c} \overset{\textsc{(Reg-Act)}}{=}
  f\reg{a}{a} \otimes (\id_I)\reg{0}{c} \overset{\textsc{(Reg-$\otimes$)}}{=}
  (f \otimes \id_I)\reg{a}{a \oplus c} \overset{\textsc{($\otimes$-U-A)}}{=}
  f\reg{a}{b}.$ 
\end{proof}

When $E$ is a "separation algebra", the witnesses of $a \leqslant b$ are necessarily unique. Moreover, when $E$ is an "effect algebra", not only do we have regrading functors $(-)\reg{a}{1}$ for every grade $a$, their behaviour is given by monoidal product with the identity on $I$ at the grade $a^\bot$, i.e. $f\reg{a}{1} = f \otimes (\id_I)\reg{0}{a^\bot}$.

  Beware of \emph{red herrings} \cite{redherring}: \pcm{}-graded monoidal categories are not, in general, "monoidal categories" equipped with extra stuff, structure, or properties.
  However, we do have the following.%

\begin{lemma} \label{idem}
  Let $\mathbb{C}$ be an "$E$-graded monoidal category", and let $e$ be an idempotent in $E$. Then the category $\mathbb{C}_e$ is a "strict monoidal category" with monoidal product $\otimes_{e,e}$ and unit $I$.
\end{lemma}
\begin{proof}
  If $e = e \oplus e$ is an idempotent, $e \perp e$ and so the monoidal product operation $(\otimes)_{e,e}$ has the type required of a "strict monoidal structure", $ℂ_e(X;Y) \times ℂ_e(X';Y') \to ℂ_e(X \otimes X'; Y \otimes Y')$.
  The axioms that these must satisfy (\Cref{moncat}) are exactly given by %
  \textsc{$\otimes$-\textsc{U-A}}, \textsc{$\otimes$-\textsc{ID}} and \textsc{Inter} of \Cref{gradmoncat}.
\end{proof}

\begin{corollary} \label{zeromon}
  For every "$(E, \oplus, 0)$-graded monoidal category", the category $ℂ_0$ at the unit of $E$ is a "monoidal category".
\end{corollary}

The \pcm{} $\bf{1}$ (\Cref{pcmone}) therefore provides our first example, which provides a sanity check for the definition of "$E$-graded monoidal category".

\begin{example} \label{ex:one}
  A $"\bf{1}"$-graded monoidal category is precisely a strict "monoidal category".
\end{example}

\begin{lemma} \label{toppremon}
  Let $\mathbb{C}$ be an $E$-graded monoidal category and $a \in E$ be a grade. Then the category $ℂ_a$ is a "strict premonoidal category" with
  $A \ltimes B = A \rtimes B := A \otimes B$ and whiskering functors
  \begin{align*}
    (A \ltimes -) \quad &:= \quad (\id_A \otimes_{0,a} -) : ℂ_a(X;Y) \to ℂ_a(A \otimes X; A \otimes Y) \\
    (- \rtimes A) \quad &:= \quad (- \otimes_{a,0} \id_A) : ℂ_a(X;Y) \to ℂ_a(X \otimes A; Y \otimes A).
  \end{align*}
\end{lemma}
\begin{proof}
  The whiskering functors are defined since $0 \oplus a = a = a \oplus 0$ by the unit laws of \pcms{}. That they satisfy the laws of "strict premonoidal categories" (\Cref{defn:premonoidal}) follows straightforwardly from the axioms \textsc{$\otimes$-U-A}, \textsc{$\otimes$-ID}, and \textsc{Inter}.
\end{proof}

Via this lemma, the \pcm{} "$\bf{2}$" (\Cref{pcm2}) supplies our first non-trivial example.

\begin{example} \label{ex:2gradeff}
  A $"\bf{2}"$-graded monoidal category comprises two categories $ℂ_0$ and $ℂ_1$, and a non-trivial regrading functor $(-)\reg{0}{1} : ℂ_0 \to ℂ_1$.
  By \Cref{toppremon,idem}, we have that $ℂ_0$ is monoidal and $ℂ_1$ is premonoidal. As we shall see in \Cref{eff-2grad}, the regrading functor is moreover premonoidal and has image in the "center" of $\mathbb{C}_1$,
  and so a $"\bf{2}"$-graded monoidal category is precisely a (strict) "effectful category", also known as a \emph{generalized Freyd category} \cite{commutativity,LEVY2003182,10.1007/3-540-48523-6_59,roman2025}. We investigate this case in more detail in \Cref{sec:effectful}.
\end{example}

Since $\bf{2} \cong \mathscr{P}(\mathbf{1})$, \Cref{ex:2gradeff} is a ``coarse grained'' instance of grading by arbitrary powerset \pcms{}. %

\begin{example} \label{ex:powersetgrad}
  Recall the powerset "separation algebra" $(\mathcal{P}(D), \uplus, \varnothing)$ from \Cref{powersetpcm}. In a $\mathscr{P}(D)$-graded monoidal category $\mathbb{C}$, morphisms are graded by subsets of a set $D$, and the monoidal product of morphisms $f \in \mathbb{C}_S$ and $g \in \mathbb{C}_T$ is defined if and only if the intersection $S \cap T$ is empty.
  This could model programs that use subsets of a given set $D$ of \emph{devices}, which may be thought of as resources corresponding to definite noun phrases, such as ``\emph{the} database'' or ``\emph{the} lock $x$''. The monoidal product of two programs that share a device is not defined; it is defined just when they use disjoint sets of devices. Sequential composition of programs that use the same set of devices again yields a program using that set.
  
  This notion of \emph{device} was introduced in our prior work \cite{earnshaw2025,earnshaw2025resourceful}, where we showed that every effectful category has an underlying \emph{effectful signature} -- in the context of the present paper we might call this informally a ``$\mathscr{P}(D)$-graded signature''. %
\end{example}

\begin{example} \label{ex:rw}
  Consider the \pcm{} $\mathsf{RW} := (\mathscr{P}(L)^2, \oplus, (\varnothing, \varnothing))$ where $L$ is a set and where
  $$(R_1, W_1) \oplus (R_2, W_2) := \begin{cases}
    (R_1 \cup R_2, W_1 \cup W_2) & \text{if}\ W_1 \cap (R_2 \cup W_2) = \varnothing = W_2 \cap (R_1 \cup W_1) \\
    \undefi & \text{otherwise}.
  \end{cases}$$

  A morphism in an $\mathsf{RW}$-graded monoidal category with grade $(R,W)$ could model a program which may read a set of memory locations $R \subseteq L$ and write a set of memory locations $W \subseteq L$.
  The side conditions for $\oplus$ enforce non-interference: read-read overlap is allowed, but any overlap involving a write is forbidden, so parallel composition is defined exactly for race-free pairs of morphisms.
\end{example}

\begin{example} \label{ex:intgrad}
  For a choice of $r \geqslant 0$, the real interval $([0,r], \dotplus, 0)$ is an "effect algebra" with the operation of bounded addition,
  $$x \dotplus y := \begin{cases}
    x + y & \text{if}\ x + y \leqslant r \\
    \undefi & \text{otherwise}
  \end{cases} \qquad\qquad x^{\bot} := r - x$$
  In an $[0,r]$-"graded monoidal category" $\mathbb{C}$, morphisms are graded by real numbers in the interval $[0,r]$, and the monoidal product of morphisms $f \in \mathbb{C}_x$ and $g \in \mathbb{C}_y$ is defined if and only if $x + y \leqslant r$. These grades might mark the amount of some finite resource, whose total amount is $r$, used by that morphism. For example, this could model programs that use some amount of a fixed bandwidth or processor capacity. The monoidal product exists only so long as together they do not exceed the bound $r$, in which case the amount of the resource used by the monoidal product is the sum of that used by the factors. The sequential composition of two programs using the same amount of the resource again uses that same amount.
\end{example}

\begin{example} \label{devquant}
  We can refine \Cref{ex:powersetgrad} by taking a product of intervals as follows. Consider a set $D$ equipped with a function $\mathfrak{b} : D \to \mathbb{R}$, thought of as specifying the maximum available quantity of each element of $D$.
  Consider the family of interval PCMs $\{ ([0, \mathfrak{b}(d)], \dotplus, 0) \}_{d \in D}$, and let $[0,\mathfrak{b}]$ denote the product of this family, as in \Cref{prodpcm}.
  
  A $[0,\mathfrak{b}]$-graded monoidal category is one in which every morphism may use a certain amount of each device, up to the bound specified by $\mathfrak{b}$.
  The monoidal product is defined, and has grade given by the pointwise sum of the quantities, only when this does not exceed the bound for any $d$. This is a quantitative refinement of  \Cref{ex:powersetgrad}. In that example, a device was either present or not, without multiplicity, and the monoidal product of two morphisms sharing a device was necessarily undefined.  
\end{example}

\begin{example} \label{ex:nmax}
  Let $(ℕ,\max,0)$ denote the total monoid of natural numbers with the maximum operation.
  Morphisms in a $(ℕ,\max,0)$-graded monoidal category are graded by natural numbers, the monoidal product exists for every pair of morphisms, and the grade of a monoidal product is the maximum of the grades of the factors.
It is tempting to interpret these grades as marking the \emph{running time} of a morphism, but for this to be the case, grades should sum under sequential composition, which they do not. However, we can interpret these grades as marking, for example, the maximum time of an atomic step in the execution of a program. The sequential composition of programs, each having maximum time of an atomic step $n$, again has maximum time of an atomic step $n$.
In \Cref{sec:globalcat}, and \Cref{ex:upb} in particular, we show how to build a category from an $(\mathbb{N}, \max, 0)$-graded monoidal category in which we can interpret the idea of grades summing under sequential composition.
\end{example}

\begin{example} \label{ex:nplus}
  Let $(\mathbb{N},+,0)$ denote the total monoid of natural numbers with addition.
  Morphisms in a $(ℕ,+,0)$-graded monoidal category are graded by a natural number, the monoidal product exists for every pair of morphisms, and the grade of a monoidal product is the sum of the grades of the factors.
  We might interpret these grades as measuring how much of some reusable ambient resource a program has access to, where running programs in parallel uses disjoint resources (and hence has grade given by the sum of the factors), but running programs sequentially uses the same resource pool (and so the grades must match, and do not accumulate). An example of such a resource might be auxiliary memory cells.
\end{example}

\begin{example}
  Let $(S, \join, \bot)$ be a join semilattice, which might for example model clearance or security levels, with the join of two levels being the least level above its factors. Since $(S, \join,\bot)$ is also a total monoid, we can consider an $(S, \join, \bot)$-graded monoidal category $\mathbb{C}$. A morphism $f \in \mathbb{C}_{\ell_1}$ where $\ell_1 \in S$, might model the fact that one needs clearance level $\ell_1$ to run the program $f$. Then given $g \in \mathbb{C}_{\ell_2}$, the grade of the monoidal product $f \otimes g$, which is always defined, is given by $\ell_1 \join \ell_2$, indicating the clearance required to run $f$ and $g$ in parallel.
\end{example}

In the following section, we shall need the category of \pcm{}-graded monoidal categories.

\begin{definition} \label{gradmorphism}
  A morphism of \pcm-graded monoidal categories $(M, \phi) : (\mathbb{C},E) \to (\mathbb{D}, F)$ comprises

  \begin{itemize}
    \item a monoid homomorphism $M : (\mathbb{C}\obj,\otimes_{\mathbb{C}}, I_{\mathbb{C}}) \to (\mathbb{D}\obj, \otimes_{\mathbb{D}}, I_{\mathbb{D}})$
    \item a "homomorphism of partial commutative monoids" $\phi : (E, \oplus_E, 0_E) \to (F, \oplus_F, 0_F)$,
    \item for every $e \in E$, a functor $M_{e} : \mathbb{C}_e \to \mathbb{D}_{\phi(e)}$ with action $M$ on objects.
    \end{itemize}
      These must satisfy the axioms enforcing preservation of monoidal products and regradings
      \begin{itemize}
      \item $M_{e \oplus e'}(g \otimes h) = M_e(g) \otimes M_{e'}(h)$, whenever $g \in \mathbb{C}_e$ and $h \in \mathbb{C}_{e'}$ have "orthogonal" grades $e \perp e'$, and
      \item $M_f\!\left(g\reg{e}{f}\right) = M_e(g)\reg{\phi(e)}{\phi(f)}.$
      \end{itemize}
      The last axiom is well typed, since morphisms of \pcms{} are necessarily monotone with respect to the "extension order".
 \end{definition}

\begin{proposition} \label{prop:indexed}
  There is a functor $(-)\text{-}\GradMon : \PCM\op \to \Cat$, taking a \pcm{} $E$ to the category of $E$-graded monoidal categories and morphisms between them.
\end{proposition}
\begin{proof}
  Let $E$ be a \pcm{}. Then $E\text{-}\GradMon$ has objects $E$-graded monoidal categories, and morphisms (\Cref{gradmorphism}) of the form $(M, \id_E)$. Define $(M, \id_E) \circ (N, \id_E)$ to have monoid homomorphism $M \circ N$, morphism of \pcms{} $\id_E$, and functors $(M \circ N)_e := M_e \circ N_e$. This is unital with identities $(\id_ℂ, \id_E)$.
  Let $\phi : E \to F$ be a morphism of \pcms{}. Define a functor $\phi^{*} : F\text{-}\GradMon \to E\text{-}\GradMon$ on objects by sending an $F$-graded monoidal category $ℂ$ to one having the same objects, and $\phi^{*}{ℂ}_e(X;Y) := ℂ_{\phi(e)}(X;Y)$. Since $\phi$ is a morphism of \pcms{}, $e \perp e'$ implies $\phi(e) \perp \phi(e')$, and so monoidal products as well as sequential composition can be defined by those in $ℂ$. On a morphism $(P, \id_F) : (ℂ,F) \to (\mathbb{D},F)$ in $F\text{-}\GradMon$, $\phi^{*}(P) := (P, \id_E)$ with local functors $P_e : \phi^{*}(ℂ)_e \to \phi^{*}(\mathbb{D})_e$ being $P_e : ℂ_{\phi(e)} \to \mathbb{D}_{\phi(e)}$ from the definition of $P$, which assignment is clearly functorial.
\end{proof}

\begin{definition} \label{egradmon}
  $\GradMon$ is the total category of the corresponding split fibration over $\PCM$ induced by \Cref{prop:indexed}.
  For a \pcm{} $E$, the category $E\text{-}\GradMon$ is a subcategory of $\GradMon$, the fibre over $E$.
\end{definition}

\begin{remark}
  The definition of \pcm{}-"graded monoidal category" is \emph{strict}, in the sense that it has a \emph{monoid} of objects, and its axioms are such that we recover \emph{strict} monoidal categories in \Cref{ex:one}, and \emph{strict} effectful categories in \Cref{ex:2gradeff}. We believe that the definition could be weakened, along the lines of asking for a monoidal category of objects. This would have the advantage of capturing more examples ``in nature'', at the price of adding considerable complexity to the definition. Monoidal categories and premonoidal categories (and therefore effectful categories) are \emph{equivalent} to their strict notions, in the sense of Mac Lane's strictification theorem \cite{power1997premonoidal,mac_lane1971categories}. We expect a similar result would hold for a weakened definition of \pcm{}-"graded monoidal category". However, we do not attempt to make this precise here.
\end{remark}

\section{Effectful categories are \bf{2}-graded monoidal categories} \label{sec:effectful}

This section lays out in detail the isomorphism between "effectful categories" and $"\bf{2}"$-graded monoidal categories (\Cref{eff-2grad}), previewed in \Cref{ex:2gradeff}, establishing a new perspective on a well established structure in the semantics of programming languages. We also show that if $E$ has a top element, for example, when it is an "effect algebra", then $E$-graded monoidal categories can be functorially ``squashed'' into $\bf{2}$-graded monoidal categories, and this is a coreflector (\Cref{coreflective}). We then introduce symmetric and cartesian structure for graded monoidal categories, and show that cartesian $\bf{2}$-graded monoidal categories are isomorphic to Freyd categories. %

\begin{lemma}\label{lem:interchange}
  Let $ℂ$ be an "$E$-graded monoidal category", and let $f \in ℂ_0(X; Y)$ and $g \in ℂ_a(X';Y')$. Then $f$ and $g$ interchange in $ℂ_a$: $(f\reg{0}{a} \otimes \id_{X'})\comp (\id_Y \otimes g) = (\id_X \otimes g) \comp (f\reg{0}{a} \otimes \id_{Y'}) = f \otimes g,$ and similarly for $g \otimes f$.
\end{lemma}
\begin{proof}
  See \Cref{app:lem:interchange}.
\end{proof}

\begin{proposition} \label{eff-2grad}
  The category $\bf{2}\text{-}\GradMon$ is isomorphic to the category $"\Eff"$ of strict effectful categories and effectful functors.
\end{proposition}
\begin{proof}
  We define a functor $F : \bf{2}\text{-}\GradMon \to "\Eff"$. Let $\mathbb{C}$ be a $\bf{2}$-graded monoidal category. We claim that $F(\mathbb{C}) := (\mathbb{C}_0, \mathbb{C}_1, (-)\reg{0}{1} : ℂ_0 \to ℂ_1)$ is an "effectful category".
  We already have that $ℂ_0$ is a "monoidal category" by \Cref{zeromon}, and $ℂ_1$ is a "premonoidal category" by \Cref{toppremon}.
  The axioms \textsc{Reg-Act} and \textsc{Reg-$\otimes$} entail that the regrading functor $(-)_0^{ 1}$ preserves whiskerings, thus we have an identity-on-objects "strict premonoidal functor" $(-)_0^{ 1} : \mathbb{C}_0 \to \mathbb{C}_1$. That the image of this functor is "central" follows from unfolding the definition of $\ltimes$ and $\rtimes$ and applying \Cref{lem:interchange} with $a = 1$. %

For the action on morphisms, recall that a morphism $(M, \id_2) : ℂ \to 𝔻$ in $\mathbf{2}\text{-}\GradMon$ comprises a monoid homomorphism $M : (ℂ\obj, \otimes_ℂ, I_ℂ) \to (𝔻\obj, \otimes_𝔻, I_𝔻)$ and functors $M_0 : ℂ_0 \to 𝔻_0$, $M_1 : ℂ_1 \to 𝔻_1$ with action on objects given by $M$, satisfying compatibility with monoidal product and regradings. Since $0 \perp 0$ and the functors $M_e$ are monoid homomorphisms on objects, the monoidal product axiom gives us that $M_0$ is a "strict monoidal functor", and similarly that $M_1$ preserves whiskerings and hence is a "strict premonoidal functor". Finally, the regrading axiom is precisely the condition that these commute with the effectful categories $F(\mathbb{C})$ and $F(\mathbb{D})$.%

Conversely, let $(\mathbb{V}, ℂ, \eta)$ be an "effectful category". We define a $\mathbf{2}$-graded monoidal category $G(\mathbb{V}, ℂ, \eta)$ by setting $G(\mathbb{V}, ℂ, \eta)_0 = \mathbb{V}$ and $G(\mathbb{V}, ℂ, \eta)_1 = \mathbb{C}$, with only non-trivial regrading functor given by $(-)_0^1 := \eta : \mathbb{V} \to \mathbb{C}$. The monoidal product $(\otimes)_{0,0}$  is the tensor of $\mathbb{V}$. For $f \in \mathbb{V}(X;Y)$ and $g \in \mathbb{C}(X';Y')$, define
$$f \otimes_{0,1} g := (\eta(f) \rtimes X') \comp (Y \ltimes g).$$
The axioms \textsc{Reg-Act} are immediate. For \textsc{Reg-$\otimes$}, the only non-trivial cases are those involving the regrading $\eta$; for $f \in \mathbb{V}(X;Y)$ and $g \in \mathbb{V}(X';Y')$ we have
\begin{align*}
  (f \otimes g)_0^1
  &:= \eta(f \otimes g) \\
  &= \eta((f \otimes \id_{X'}) \comp (\id_Y \otimes g))
    &&\text{($\mathbb{V}$ mon. cat)} \\
  &= \eta(f \otimes \id_{X'}) \comp \eta(\id_Y \otimes g)
      &&\text{($\eta$ func.)} \\
  &= (\eta(f) \rtimes X') \comp (Y \ltimes \eta(g))
    &&\text{($\eta$ premon. cat)} \\
  &=: f \otimes_{0,1} \eta(g),
\end{align*}
and similarly $\eta(f \otimes g) = \eta(f) \otimes_{1,0} g$. The axioms $\otimes$-\textsc{U-A} and $\otimes$-\textsc{ID} follow from the strict premonoidal axioms in $\mathbb{C}$ together with preservation of identities by $\eta$.

For \textsc{Inter}, the grade $0$ case is interchange in $\mathbb{V}$. For the mixed case, let $f,h \in \mathbb{V}$ and $g,k \in \mathbb{C}$ be composable. Then
\begin{align*}
  (f \otimes_{0,1} g)\comp(h \otimes_{0,1} k)
  &:= (\eta(f) \rtimes X') \comp (Y \ltimes g) \comp (\eta(h) \rtimes Y') \comp (Z \ltimes k) \\
  &= (\eta(f) \rtimes X') \comp (\eta(h) \rtimes X')\comp(Z \ltimes g) \comp (Z \ltimes k)
        &&\text{($\eta$ central)} \\
  &= (\eta(f \comp h) \rtimes X') \comp (Z \ltimes (g \comp k))
    &&\text{($\eta,\rtimes,\ltimes$ func.)} \\
  &=: (f \comp h)\otimes_{0,1}(g \comp k).
\end{align*}
The case of $\otimes_{1,0}$ is analogous. Thus $G(\mathbb{V}, ℂ, \eta)$ is a $\mathbf{2}$-graded monoidal category.

An "effectful functor" $(H_0,H_1) : (\mathbb{V}, ℂ, \eta) \to (\mathbb{V}', ℂ', \eta')$ determines a morphism $G(H_0,H_1)$ in $\mathbf{2}\text{-}\GradMon$ with monoid homomorphism given by the action on objects of $H_0$ (which necessarily coincides with that of $H_1$) and local functors $G(H_0,H_1)_0$ and $G(H_0,H_1)_1$ given by $H_0$ and $H_1$ respectively. The fact that $H_0$ is strict monoidal and $H_1$ is strict premonoidal, along with the condition $\eta' \circ H_0 = H_1 \circ \eta$ gives the required preservation of monoidal products in the grade combinations where they are defined. The condition $\eta' \circ H_0 = H_1 \circ \eta$ corresponds precisely to the regrading compatibility axiom.
Thus we recover a morphism of $\mathbf{2}$-graded monoidal categories. It is clear that these processes are mutually inverse.
\end{proof}

\begin{definition}
  Denote by $""\GradMonTop""$ the subcategory of $\GradMon$ whose grading \pcms{} have a "top element", and whose morphisms preserve the top grade.
\end{definition}

Since $"\bf{2}"$ has a top element and morphisms in $\TwoGradMon$ preserve it, $\TwoGradMon$ is a full subcategory of $"\GradMonTop"$, and moreover a coreflective one.

\begin{restatable}[]{theorem}{coreflective} \label{coreflective}
  The full subcategory inclusion $i : \TwoGradMon \hookrightarrow "\GradMonTop"$ has a right adjoint. That is, $\TwoGradMon \cong \Eff$ is a coreflective subcategory of $"\GradMonTop"$.
\end{restatable}
\begin{proof}
  We construct a universal morphism from the functor $i$ to an arbitrary $ℂ$ in $"\GradMonTop"$ with monoid of objects $(ℂ\obj, \otimes ,I)$, graded by a \pcm{} $E$ with a top element. Define $Rℂ$ to be the $"\bf{2}"$"-graded monoidal category" over the same monoid of objects and where, $Rℂ_0(X; Y) := ℂ_0(X; Y)$ and $Rℂ_1(X; Y) := ℂ_\top(X; Y)$, the non-trivial regrading $(-)\reg{0}{1}$ is $(-)\reg{0}{\top}$ and monoidal product for two morphisms of grade 0, or of grade 0 and grade 1 is just that in $ℂ$. Define a morphism $\varepsilon_ℂ : i(Rℂ) \to ℂ$ in $"\GradMonTop"$ to be identity-on-objects, with morphism of \pcms{} $2 \to E$ the unique $\top$ preserving morphism, sending $0$ to $0$ and $1$ to $\top$. The required functors from grade $0$ to $0$ and grade $1$ to $\top$ are then simply identities, and it is immediate that these commute with regradings and monoidal products where defined.
  Let $𝔻$ be a $2$-graded monoidal category and $(M,\phi) : i𝔻 \to ℂ$ a morphism in $"\GradMonTop"$, where $\phi$ is necessarily the unique top-preserving \pcm{} morphism $2 \to E$. Define $\widehat{M} : 𝔻 \to Rℂ$ in $\TwoGradMon$ to have the same action on objects and morphisms as $M$. Then it is easy to verify that $\varepsilon_ℂ \circ i(\widehat{M}) = M$, and moreover $\widehat{M}$ is unique with this property: if $\widehat{N}$ is another morphism such that $\varepsilon_ℂ \circ i(\widehat{N}) = M$ then $\widehat{N} = \widehat{M}$. Therefore $i$ has a right adjoint given on objects by $R$ and on a morphism $F : 𝔻 \to ℂ$ by the unique $\widehat{F \circ \varepsilon_𝔻}$.
\end{proof}

\subsection{Symmetric and cartesian structure} \label{sec:symcart}

Monoidal categories lack structure corresponding to the syntactic rules of symmetry, weakening, and contraction. This makes them suited to the semantics of non-cartesian domains, such as quantum \cite{SELINGER2004} or probabilistic programming languages \cite{stein2021structural,dilavore2025simple}. For classical programming languages, we must supply \emph{cartesian} structure.
\emph{Freyd categories} \cite{LEVY2003182,10.1007/3-540-48523-6_59} are the cartesian cousins of effectful categories. Their category of pure morphisms is a \emph{cartesian} monoidal category, and the category of effectful morphisms is a \emph{symmetric} premonoidal category. In this section, we define \emph{cartesian} \pcm{}-graded monoidal categories, and show that cartesian $"\bf{2}"$-graded monoidal categories are precisely Freyd categories, extending \Cref{eff-2grad}. We shall begin with symmetric structure.

\begin{definition}
  Let $(E,\oplus,0)$ be a partial commutative monoid, and let $\mathbb{C}$ be an $E$-graded monoidal category. We say that $\mathbb{C}$ is \emph{symmetric} in case the monoidal category $\mathbb{C}_0$ is a "symmetric strict monoidal category" with braidings $\sigma_{X,Y} \in \mathbb{C}_0(X \otimes Y ; Y \otimes X)$ such that for all $a,b \in E$ with $a \perp b$, and all $f \in \mathbb{C}_a(X;Y)$ and $g \in \mathbb{C}_b(X';Y')$, we have $$(f \otimes g)\comp
  (\sigma_{Y,Y'})_0^{a \oplus b} = (\sigma_{X,X'})_0^{a \oplus b}\comp(g \otimes f).$$
\end{definition}

\begin{example} \label{sym1}
  A symmetric $\bf{1}$-graded monoidal category is simply a "symmetric strict monoidal category". More generally, in a symmetric $E$-graded monoidal category $\mathbb{C}$, the category $\mathbb{C}_a$ is symmetric monoidal for any idempotent $a \in E$, with braidings $(\sigma_{X,Y})\reg{0}{a}$. The straightforward proof of this fact is essentially that of the following lemma.
\end{example}

\begin{lemma} \label{symtop}
  Let $E$ be a \pcm{}, let $\mathbb{C}$ be a symmetric $E$-graded monoidal category, and let $a \in E$ be a grade. Then $\mathbb{C}_a$ is a "symmetric premonoidal category" with braidings $(\sigma_{X,Y})_0^a$.
\end{lemma}
\begin{proof}
See \Cref{app:symtop}.
\end{proof}

\begin{proposition}\label{symeff}
  Let $E$ be a \pcm{}, let $\mathbb{C}$ be a symmetric $E$-graded monoidal category and let $a \in E$ be a grade. Then there is a "symmetric effectful category" given by $(-)_0^a : \mathbb{C}_0 \to \mathbb{C}_a$.
\end{proposition}
\begin{proof}
  We have that $\mathbb{C}_a$ is symmetric premonoidal from \Cref{symtop}, and $(-)_0^a$ preserves the braiding by construction. That $(-)_0^a$ is strict premonoidal and has central image follows by the same reasoning as in the first part of the proof of \Cref{eff-2grad}.
\end{proof}

\begin{definition}
  A morphism $M$ of symmetric \pcm{}-graded monoidal categories is a morphism as in \Cref{gradmorphism}, for which $M_0(\sigma_{X,Y}) = \sigma_{MX,MY}$, i.e. $M_0$ is a "symmetric strict monoidal functor".
\end{definition}

\begin{proposition} \label{symeff-sym2grad}
  The category $"\SymEff"$ is isomorphic to the category $\TwoSymGradMon$.
\end{proposition}
\begin{proof}
  From \Cref{symeff} with $a=1$, a symmetric $"\bf{2}"$-graded monoidal category induces a "symmetric effectful category".
  From \Cref{eff-2grad} we have that an effectful category induces a $"\bf{2}"$-graded monoidal category.
  Assume $(\mathbb{V},\mathbb{C},\eta)$ is symmetric.
  To show the induced $"\bf{2}"$-graded monoidal category is symmetric, we must check
  $(f \otimes g)\comp (\sigma_{Y,Y'})_0^{a\oplus b} = (\sigma_{X,X'})_0^{a\oplus b}\comp (g \otimes f)$
  for all $a \perp b$.
  For $a=b=0$, this follows from the fact that $\mathbb{V}$ is symmetric.
  For $a=0, b=1$, we have $f\reg{0}{1}:=\eta(f)$ and $(\sigma_{X,X'})_0^1:=\eta(\sigma_{X,X'})$, so we must show
  $(\eta(f) \otimes g)\comp \sigma_{Y,Y'} = \sigma_{X,X'}\comp (g \otimes \eta(f)),$
  which follows from the fact that $\mathbb{C}$ is "symmetric premonoidal", using that $\eta$ preserves braidings and lands in the center. The case $a=1,b=0$ is analogous.
  Given a morphism of symmetric $"\bf{2}"$-graded monoidal categories, \Cref{eff-2grad} gives us an "effectful functor" with components $M_0$ and $M_1$. The condition that the morphism is symmetric is precisely that $M_0$ is a "symmetric strict monoidal functor", so it remains to check that $M_1$ is a "symmetric premonoidal functor", which follows from $M_1((\sigma_{X,Y})\reg{0}{1}) = M_0(\sigma_{X,Y})\reg{0}{1} = (\sigma_{MX,MY})\reg{0}{1}.$ Conversely, given a "symmetric strict effectful functor", we have a morphism of $"\bf{2}"$-graded monoidal categories from \Cref{eff-2grad}, and the condition that $M_0$ preserve symmetries is the extra condition that the functor between the monoidal categories be symmetric strict.
\end{proof}

\begin{proposition} \label{coreflective2}
The adjunction of \Cref{coreflective} restricts to symmetric structures, i.e. $\TwoSymGradMon \cong \SymEff$ is a coreflective subcategory of $"\SymGradMonTop"$.
\end{proposition}
\begin{proof}
  Let $ℂ$ be a symmetric $E$-graded monoidal category for a \pcm{} $E$ with top element. From \Cref{coreflective} we obtain a $"\bf{2}"$-graded monoidal category $Rℂ$. Since by definition $(Rℂ)_0 = ℂ_0$, and $(-)\reg{0}{1}$ preserves braidings, $Rℂ$ is a symmetric $"\bf{2}"$-graded monoidal category. Since the counit is identity on objects and morphisms, it preserves the braiding and so is a morphism of symmetric $E$-graded monoidal categories. Given a symmetric $\bf{2}$-graded monoidal category $\mathbb{D}$ and a symmetric morphism $(M, \phi) : i\mathbb{D} \to ℂ$ in $\SymGradMonTop$, we have the unique $\widehat{M}$ as in \Cref{coreflective}, which we must check is a morphism of $\TwoSymGradMon$. It suffices to check that it preserves braidings, which follows from $\widehat{M_0}(\sigma_{X,Y}) := M_0(\sigma_{X,Y}) = \sigma_{MX,MY}$.
\end{proof}

\begin{definition}
 Let $E = (E,\oplus,0)$ be a \pcm{}. An $E$-graded monoidal category is said to be \emph{cartesian} in case $\mathbb{C}_0$ is a "cartesian monoidal category", and the braiding there makes $\mathbb{C}$ into a symmetric $E$-graded monoidal category.
\end{definition}

\begin{proposition} \label{cart-freyd}
  Let $E$ be a \pcm{}, let $\mathbb{C}$ be a cartesian $E$-graded monoidal category, and let $a \in E$ be a grade. Then there is a "Freyd category" given by $(-)_0^a : \mathbb{C}_0 \to \mathbb{C}_a$.
\end{proposition}
\begin{proof}
  By \Cref{symeff}, $(-)_0^a : \mathbb{C}_0 \to \mathbb{C}_a$ is a "symmetric effectful category" and since $\mathbb{C}_0$ is cartesian monoidal by definition, this is a "Freyd category".
\end{proof}

\begin{definition}
  A morphism of cartesian $E$-graded monoidal categories is a morphism as in \Cref{gradmorphism}, for which $M_0$ is a "cartesian monoidal functor".
\end{definition}

\begin{theorem}
  The category of "Freyd categories" is isomorphic to the category of cartesian $"\bf{2}"$-graded monoidal categories, $\TwoCartGradMon \cong \Freyd$.
\end{theorem}
\begin{proof}
  By \Cref{symeff-sym2grad}, symmetric $"\bf{2}"$-graded monoidal categories are isomorphic to symmetric effectful categories. We show that this isomorphism restricts to the cartesian case. On objects, if $\mathbb{C}$ is a cartesian $"\bf{2}"$-graded monoidal category, then by \Cref{cart-freyd} the induced effectful category $(-)_0^1 : \mathbb{C}_0 \to \mathbb{C}_1$ is a Freyd category since $\mathbb{C}_0$ is cartesian by definition.

Conversely, if $(\mathbb{V},\mathbb{C},J)$ is a Freyd category, then \Cref{symeff-sym2grad} yields a symmetric $"\bf{2}"$-graded monoidal category with grade $0$ category $\mathbb{V}$ and grade $1$ category $\mathbb{C}$. Since $\mathbb{V}$ is cartesian monoidal, this is in fact a cartesian $"\bf{2}"$-graded monoidal category.
On morphisms, the isomorphism of \Cref{symeff-sym2grad} sends an effectful functor $(M_0,M_1)$ to the corresponding morphism of symmetric $"\bf{2}"$-graded monoidal categories with the same components. The additional requirement for a morphism in $\Freyd$ is precisely that $M_0$ be cartesian, which is exactly the additional requirement for a morphism in $\TwoCartGradMon$. %
  \end{proof}

\begin{proposition} \label{twocartgradmonco}
The adjunction of \Cref{coreflective2} restricts to cartesian structures, i.e. $\TwoCartGradMon \cong \Freyd$ is a coreflective subcategory of $"\CartGradMonTop"$.
\end{proposition}
\begin{proof}
  We extend \Cref{coreflective2}. Let $ℂ$ be a cartesian $E$-graded monoidal category for a \pcm{} $E$ with top element. From \Cref{coreflective2} we obtain a symmetric $\bf{2}$-graded monoidal category, $Rℂ$. Since by definition $(Rℂ)_0 = ℂ_0$, it is a cartesian $\bf{2}$-graded monoidal category. Since the counit is the identity on objects and morphisms, it is a morphism of cartesian $E$-graded monoidal categories. Given a cartesian $\bf{2}$-graded monoidal category $\mathbb{D}$ and a morphism $(M,\phi) : i \mathbb{D} \to ℂ$ in $\CartGradMonTop$, we have from \Cref{coreflective2} a unique $\widehat{M} : \mathbb{D} \to R\mathbb{C}$ in $\TwoSymGradMon$. Since $\widehat{M}_0$ is defined to be exactly $M_0$, which is a cartesian monoidal functor by assumption, $\widehat{M}$ preserves cartesian structure and hence is a morphism in $\TwoCartGradMon$.
\end{proof}

Jeffrey \cite{jeffrey1997premonoidal1} considered programming language semantics in \emph{triples} of a cartesian category, a symmetric monoidal category and a symmetric premonoidal category, corresponding respectively to values, pure computations and effectful computations.
From this point of view, the \pcm{} $"\mathbf{2}"$ is a ``truncated'' instance of the total monoid given by \emph{maximum}, which extends to three elements as follows.

\begin{example} \label{pcm3}
  Denote by $""\bf{3}""$ the \pcm{} with three elements $\{0,1,2\}$ and partial operation
  $$x \oplus y := \begin{cases} \undefi & x = y = 2 \\
    \max(x,y) & \text{otherwise.}
    \end{cases}$$
\end{example}

Then combining the above results, we observe that

\begin{corollary}
  Cartesian $"\mathbf{3}"$-graded monoidal categories are the triples of Jeffrey \cite{jeffrey1997premonoidal1}.
\end{corollary}

\section{Categories from PCM-graded monoidal categories}\label{sec:globalcat}

Although the sequential composition operators of an "$E$-graded monoidal category" $\mathbb{C}$ are \emph{homogeneous} in the grade, if the "extension order" of $E$ is appropriately well behaved, we can make sense of the sequential composition of \emph{heterogeneously} graded morphisms, assembling the ``local'' categories $\mathbb{C}_a$ into a single ``global'' category. The simplest case is when the "extension order" of $E$ has binary joins, in which case we can take the disjoint union of the hom-sets.

\begin{proposition}\label{lem:global-category-structure}
  Let $(E,\oplus,0)$ be a \pcm{} whose "extension order" is a poset with binary joins, and let $ℂ$ be an "$E$-graded monoidal category".
  There is a category, also denoted by $ℂ$, with objects $ℂ\obj$ and hom-sets $ℂ(X;Y) := \coprod_{a \in E} ℂ_a(X;Y)$.
  Identities are given by $\id_X$ (grade 0 identities), and composition for $f \in ℂ_a(X;Y)$ and $g \in ℂ_b(Y;Z)$ by $f ; g := f\reg{a}{a \join b} \comp g\reg{b}{a \join b}.$
\end{proposition}
\begin{proof}
Let $f\in ℂ_a(X;Y)$, $g \in ℂ_b(Y;Z)$, and $h\in ℂ_c(Z;W)$. For associativity, we have by definition
$$(f ; g) ;  h = (f\reg{a}{a\join b} \comp g\reg{b}{a\join b})\reg{a\join b}{(a\join b)\join c} \comp h_c^{(a\join b)\join c}.$$
Then we have
\begin{align*}
  &(f\reg{a}{a\join b}\comp g\reg{b}{a\join b})\reg{a\join b}{(a\join b)\join c} \comp h\reg{c}{(a\join b)\join c} \\
  &= ((f\reg{a}{a\join b})\reg{a\join b}{a \join b \join c} \comp (g_b^{a\join b})\reg{a \join b}{a \join b \join c})  \comp h\reg{c}{(a\join b)\join c}
  &&\text{(regrad. func.)} \\
  &= (f\reg{a}{a \join b \join c} \comp g\reg{b}{a \join b \join c}) \comp h\reg{c}{(a\join b)\join c}
  &&(\textsc{Reg-Act}) \\
  &= f\reg{a}{a \join b \join c} \comp (g\reg{b}{a \join b \join c} \comp h\reg{c}{(a\join b)\join c}),
  &&(\comp\ \text{assoc.})
\end{align*}
and the same reasoning starting with $f ; (g ;  h)$ arrives at the same term.
    For left unitality, let $f \in \mathbb{C}_a(X;Y)$ then we have
    $\id_X ; f
      := \id_X\reg{0}{0 \join a} \comp f\reg{a}{0 \join a}
      =  f\reg{a}{a}
      = f$,
  and similarly for right unitality.
\end{proof}

\begin{example} \label{ex:pcmg}
  The "extension order" of a powerset \pcm{} (\Cref{powersetpcm}) is exactly the usual subset inclusion preorder, since $U = T - S$ witnesses $S \leqslant T$. Therefore join is given by the union $S \cup T$.

  Recall from \Cref{ex:powersetgrad} the \emph{device} interpretation of $\mathscr{P}(X)$-graded monoidal categories: a morphism of grade $S \subseteq X$ is a program that uses devices in $S$, and parallel composition is defined only for disjoint device sets. Under this interpretation, sequential composition in the category $ℂ$ corresponds to the intuitive notion that a sequential program uses the union of the devices appearing in each term of the sequence.
\end{example}

\begin{example} \label{ex:totn}
  The (total) commutative monoid $(\mathbb{N},+,0)$ has binary joins, $n \join m := \max(n,m)$.
  Under the interpretation of $(\mathbb{N},+,0)$-graded monoidal categories as those modelling programs whose grade corresponds to the amount of a reusable resource used by that program (such as auxiliary ``scratch'' memory cells), composition in the resulting global category captures the intuitive idea that we should be able to compose an $n$ graded and an $m$ graded program to obtain a $\max(n,m)$ graded one.
\end{example}

In fact, inspection of the proof of \Cref{lem:global-category-structure} shows that we do not need the full strength of a join: it suffices that $\vee$ is an ``upper-bounding monoid'' structure in the following sense.

\begin{proposition} \label{upperboundingglobal}
  Let $(E, \oplus, 0)$ be a \pcm{} and let $\vee$ be an associative binary operation on $E$, with unit $0$, and such that it provides an upper bound of its arguments in the "extension order", $a \leqslant a \vee b$ and $b \leqslant a \vee b$. Then an $E$-graded monoidal category gives a category $ℂ$ exactly as defined in \Cref{lem:global-category-structure}. When $\vee$ is moreover idempotent, the composition of two morphisms $f \in \mathbb{C}_a(X;Y)$ and $g \in \mathbb{C}_a(Y;Z)$ with the same grade $a$ coincides with the operation $(\comp_a)$ of the $E$-graded monoidal category.
\end{proposition}

Any total commutative monoid provides an example of such an upper-bounding operation, given by the operation of the monoid.
For example $+$ is an upper bounding operation for the total monoid $(\mathbb{N},+,0)$, differing from the join, given by $\max$.
We can also swap these operations, as in the following example.

\begin{example}\label{ex:upb}
  Consider the total commutative monoid of natural numbers with maximum, $(\mathbb{N}, \max, 0)$. 
  The "extension order" has a join given by $\max$ but addition is an upper-bounding operation in the sense of \Cref{upperboundingglobal}.
  Therefore, given an $(\mathbb{N},\max,0)$-graded monoidal category (\Cref{ex:nmax}) we obtain a category from \Cref{lem:global-category-structure} in which grades \emph{sum} under sequential composition: this might model execution \emph{time} or other accumulating \emph{costs}.
\end{example}

The category $ℂ$ constructed above contains copies of ``the same'' morphism in different grades, being a \emph{disjoint} union of the hom-sets at each grade, $\mathbb{C}_a$. It is also natural to consider identifying morphisms along regrading functors, which allows us to weaken the hypothesis on the "extension order" to directedness.

\begin{proposition} \label{directedc}
  Let $(E, \oplus, 0)$ be a \pcm{} whose "extension order" is directed. Then there is a category $\overline{ℂ}$ with objects $ℂ\obj$ and hom-sets $\overline{ℂ}(X;Y) := \coprod_{a} ℂ_a(X;Y)/{\equiv}$
  where $\equiv$ is the least equivalence relation generated by pairs $\langle a, f \rangle \equiv \langle b, f\reg{a}{b} \rangle$ for every $a \leqslant b$ in $E$.
\end{proposition}
\begin{proof}
  See \Cref{app:directedc}.
\end{proof}

\begin{example}
  \Pcms{} whose extension orders have binary joins, or have an upper-bounding monoid structure (in the sense of \Cref{upperboundingglobal}), are directed, and so \Cref{ex:pcmg,ex:totn,ex:upb} also give rise to categories as in \Cref{directedc}.
\end{example}

In case $E$ has a top element in its "extension order", for example, when $E$ is an "effect algebra", this global category $\overline{ℂ}$ is isomorphic to $ℂ_\top$, and in case $\oplus$ is total, $\overline{ℂ}$ is a strict monoidal category.

\begin{proposition}
Let $(E, \oplus, 0)$ be a \pcm{} whose "extension order" has a top element. Then the category $\overline{ℂ}$ with objects $ℂ\obj$ and hom-sets $\overline{ℂ}(X;Y) := \coprod_{a} ℂ_a(X;Y)/{\equiv}$ defined in \Cref{directedc} exists and is isomorphic to the category $ℂ_\top$.
\end{proposition}
\begin{proof}
  $\overline{ℂ}$ exists since a "top element" implies directedness. For every $a \in E$ we have $a \leqslant \top$, so every equivalence class $[\langle a,f \rangle]$ has the representative $\langle \top, f_a^\top \rangle$.
   The map $[\langle a,f\rangle] \mapsto f_a^\top$ is then a well defined bijection $\overline{ℂ}(X;Y) \cong ℂ_\top(X;Y)$ with inverse $u\mapsto [\langle \top,u \rangle]$. These functions preserve identities since $[\langle 0,\id_X \rangle] \mapsto (\id_X)_0^\top$, and composition since for composable $[\langle a,f \rangle]$ and $[\langle b,g \rangle]$, the top element is an upper bound of $a$ and $b$, so by the definition of composition in $\overline{ℂ}$,
   \[
     [\langle a,f \rangle] \comp [\langle b,g \rangle] = [\langle \top, f_a^\top \comp_\top g_b^\top \rangle] \mapsto (f_a^\top \comp_\top g_b^\top)_\top^\top = f_a^\top \comp_\top g_b^\top.
   \]
   Hence these bijections assemble into an isomorphism of categories $\overline{ℂ} \cong ℂ_\top$.
\end{proof}
 
\begin{proposition} \label{globmon}
  When $(E,\oplus,0)$ is a (total) commutative monoid then $\overline{ℂ}$ is defined and has a strict "monoidal structure".
\end{proposition}
\begin{proof}
See \Cref{app:globmon}.
\end{proof}

\section{PCM-graded monoidal categories as monoids} \label{sec:gradprom}
In this section, we present a category-theoretic perspective on our central definition (\Cref{gradmoncat}). This reformulation better connects our notion to the existing literature, and opens it up to generalization.

Categories of grades are often treated as thin monoidal categories, for example in the literature on graded monads \cite{katsumata2014parametric,mellies2017,orchard2014semantic} and locally graded categories \cite{wood1976indicial,levy2019locally,10.1007/978-3-031-16912-0_4}. The preorder structure in such categories of grades induces regrading maps, and the monoidal structure captures the \emph{monoid} structure on grades.

To deal with the fact that our grades combine only partially, we introduce a thin \emph{promonoidal} category of grades \cite{day1970closed}. Promonoidal structure suffices to obtain a convolution monoidal structure on presheaves, which are duoidal with the pointwise cartesian monoidal product. This allows us to use the results of Heunen and Sigal \cite{heunen2023duoidally} to characterize PCM-graded monoidal categories as monoids in a category of lax monoidal functors. We shall not need promonoidal categories in their full generality, but rather the simpler case of $\Bool$-enriched promonoidal categories.

\begin{definition} \label{def:boolprof}
  Let $\Bool$ be the poset of truth values $\{\varnothing \leqslant \top\}$.
  A $\Bool$-profunctor $P : \mathbb{C} \lto \mathbb{D}$ is a functor $P : \mathbb{C}\op \times \mathbb{D} \to \Bool$.
  Composition is relational composition,
  $(Q \circ P)(c;d) := \exists\ e \in \mathbb{D},\ P(c;e) \wedge Q(e;d).$
\end{definition}

\begin{definition} \label{def:boolprom}
  A $\Bool$-""promonoidal category"" $(\mathbb{C},P,I)$ is a category $ℂ$ equipped with $\Bool$-profunctors
  $P : ℂ \times ℂ \lto ℂ$ and $I : 1 \lto ℂ,$
  satisfying the associativity and unitality laws
  $$P \circ (P \times \bid_ℂ) = P \circ (\bid_ℂ \times P),\quad P \circ (\bid_ℂ \times I) = \bid_ℂ, \quad P \circ (I \times \bid_ℂ) = \bid_ℂ,$$
  where
  $\bid_{\mathbb{C}}$ is the $\Bool$-profunctor given by change of base of $\mathsf{Hom}_\mathbb{C}$ along
  $$\Set \to \Bool := \begin{cases} \varnothing \mapsto \varnothing, \quad A \mapsto \top \quad \text{if}\ A\ \text{is non-empty}. \end{cases}$$
\end{definition}

We encode a \pcm{} as thin $\Bool$-"promonoidal category" as follows.

\begin{proposition} \label{pcmprom}
  Let $(E,\oplus,0)$ be a \pcm{}. Then $\mathbf{E} = ((E,\leqslant), P, I)$ is a $\Bool$-"promonoidal category" on the "extension order" of $E$ where
  \[
    P(a,b ; c) :=
    \begin{cases}
      \top & \text{if}\ a \perp b\ \text{and}\ a \oplus b \leqslant c,\\
      \varnothing & \text{otherwise,}
    \end{cases}
    \qquad
    I(c) := \top\ \text{for all}\ c.
  \]
\end{proposition}
\begin{proof}
  See \Cref{app:pcmprom}.
\end{proof}

A functor $F : \mathbf{E} \to \Set$, or \emph{copresheaf}, comprises a family of sets $\{F(e)\}_{e \in E}$ indexed by $E$, together with functions $F(e \leqslant e') : F(e) \to F(e')$ for each $e \leqslant e'$ in the "extension order". These families will give graded hom-sets, and the functions $F(e \leqslant e')$ will give regrading maps. We shall also need two monoidal structures on the category $[\mathbf{E},\Set]$ of functors $\mathbf{E} \to \Set$ and natural transformations between them.

\begin{proposition} \label{boolconv}
  Let $(\mathbf{E},P,I)$ be the $\Bool$-"promonoidal category" generated by a \pcm{} (\Cref{pcmprom}). Then the category $[\mathbf{E},\Set]$ has convolution monoidal structure $([\mathbf{E},\Set], \ast, J)$ where
    $$(F \ast G)(c) := \left( \coprod_{a \oplus b \leqslant c} F(a) \times G(b) \right) / \sim$$
    and $\sim$ is the least equivalence relation generated by pairs
    $(a,b,x,y) \sim (a',b',F(a \leqslant a')(x),G(b \leqslant b')(y))$
  for $x \in F(a), y \in G(b)$, $a \leqslant a'$, $b \leqslant b'$ and $a' \oplus b' \leqslant c$.
  The unit is given by $J(c) := \{\bullet\}$ for all $c$.
\end{proposition}
\begin{proof}
This is the $\Bool$-enriched case of the standard convolution monoidal structure for presheaves on a promonoidal category \cite{day1970closed}, specialized to the promonoidal category $\mathbf{E}$.
\end{proof}

Along with the pointwise cartesian monoidal structure $(\times, K)$, where $(F \times G)(e) := F(e) \times G(e)$, and $K$ is constant at the singleton, it is standard that $([\mathbf{E},\Set], \ast, J, \times, K)$ is a duoidal category \cite{street12:linking}. Therefore we have the following, where although $\mathbb{C}\obj$ is discrete, and so $\mathbb{C}\obj\op \cong \mathbb{C}\obj$, we write $\mathbb{C}\obj\op \times \mathbb{C}\obj$ to anticipate future generalization.

\begin{proposition} \label{circl} \textup{(Heunen and Sigal, \cite[\S 5, Propositions 3,4 and 5]{heunen2023duoidally})} Let $(\mathbb{C}\obj, \otimes, I)$ be a monoid, seen as a monoidal discrete category, so that $\mathbb{C}\obj\op \times \mathbb{C}\obj$ is monoidal with $(X,Y) \otimes (X',Y') := (X \otimes X', Y \otimes Y')$. The functor category
  $$\MonCatl(\mathbb{C}\obj\op \times \mathbb{C}\obj, ([\mathbf{E},\Set], \ast, J))$$ has a monoidal structure $(\circ, L)$ given by lifting $(\times,K)$,
  \begin{gather*}
      (P \circ Q)(X;Z) := \coprod_{Y \in \mathbb{C}\obj} Q(X;Y) \times P(Y;Z) \qquad
      L(X;Y) := \begin{cases}
    K & \text{if}\ X = Y \\
    e \mapsto \varnothing & \text{otherwise.}
    \end{cases}
  \end{gather*}
\end{proposition}

\begin{theorem} \label{egradismon}
  Let $(E, \oplus, 0)$ be a \pcm{}, $\mathbf{E}$ be the corresponding $\Bool$-"promonoidal category", as in \Cref{pcmprom}, and $(\mathbb{C}\obj, \otimes, I)$ be a monoid, seen as monoidal discrete category.
  An $E$-graded monoidal category with monoid of objects $(\mathbb{C}\obj, \otimes, I)$ is precisely a monoid in the monoidal category $$(\MonCatl(\mathbb{C}\obj\op \times \mathbb{C}\obj, ([\mathbf{E},\Set], \ast, J)), \circ, L),$$
  that is, a duoidally $[\mathbf{E},\Set]$-enriched Freyd category on $(\mathbb{C}\obj, \otimes, I)$, in the terminology of Heunen and Sigal \cite{heunen2023duoidally}.
\end{theorem}
\begin{proof}  
  This is mostly a case of unfolding definitions, so let us first unpack what such a monoid comprises in elementary terms. We will then examine the apparent discrepancies.
  \begin{enumerate}[label=\arabic*., ref=\arabic*]
  \item A lax monoidal functor $\mathbb{C} : \mathbb{C}\op\obj \times \mathbb{C}\obj \to ([\mathbf{E},\Set], *, J)$, which is
    \begin{enumerate}[label=(\roman*), ref=\theenumi.(\roman*)]
    \item\label{itm:lax-functor-homset} a set $\mathbb{C}_a(X;Y)$, for each pair of objects $X,Y$ in $\mathbb{C}\obj$, and grade $a \in \mathbf{E}$, and
    \item\label{itm:lax-functor-regrade} a function $(-)_a^b : \mathbb{C}_a(X;Y) \to \mathbb{C}_b(X;Y)$, for grades $a \leqslant b$ in $\mathbf{E}$,
    \item\label{itm:lax-functor-regrade-laws} such that $(-)_a^a$ is the identity, and $(-)_b^c \circ (-)_a^b = (-)_a^c$,
    \item\label{itm:lax-tensor} functions $\otimes_{a,b;c} : \mathbb{C}_a(X;Y) \times \mathbb{C}_b(X';Y') \to \mathbb{C}_c(X \otimes X'; Y \otimes Y')$ for each $a \oplus b \leqslant c$
    \item\label{itm:lax-unit} an element $\eta_a \in \mathbb{C}_a(I;I)$, for each $a \in \mathbf{E}$,
      \par\smallskip
      \noindent\hspace*{-\leftmargin} satisfying the following naturality and coherence equations
      \item\label{itm:equiv-compat} compatibility of $\otimes_{a,b;c}$ with the equivalence relation in \Cref{boolconv}, i.e. whenever $a \leqslant a'$, $b \leqslant b'$, and $a' \oplus b' \leqslant c$, $(f\reg{a}{a'}) \otimes_{a',b';c} (g\reg{b}{b'}) = f \otimes_{a,b;c} g,$
    \item\label{itm:unit-nat} $\eta_{c'} = (\eta_c)\reg{c}{c'}$ whenever $c \leqslant c'$,
    \item\label{itm:tensor-nat} $(f \otimes_{a,b;c} g)\reg{c}{d} = f \otimes_{a,b;d} g$ whenever $a \oplus b \leqslant c \leqslant d$,
    \item\label{itm:lax-laws-assoc} $(f \otimes_{a,b;x} g) \otimes_{x,c;d} h = f \otimes_{a,y;d} (g \otimes_{b,c;y} h)$ whenever $a \oplus b \leqslant x$, $x \oplus c \leqslant d$, $b \oplus c \leqslant y$, and $a \oplus y \leqslant d$,
    \item\label{itm:lax-laws-unit} $f \otimes_{a,b;c} \eta_b = f\reg{a}{c}$ and $\eta_a \otimes_{a,b;c} f = f\reg{b}{c}$ whenever $a \oplus b \leqslant c$.
    \end{enumerate}
  \item A monoidal natural transformation $(\comp) : \mathbb{C} \circ \mathbb{C} \Rightarrow \mathbb{C}$, which is
      \begin{enumerate}[label=(\roman*), ref=\theenumi.(\roman*)]
        \item\label{itm:comp-map} a function $(\comp)_a : \mathbb{C}_a(X;Y) \times \mathbb{C}_a(Y;Z) \to \mathbb{C}_a(X;Z),$ for each grade $a \in \mathbf{E}$ and objects $X,Y,Z$, 
        \item\label{itm:comp-nat} natural in the grade, i.e. for $a \leqslant b$, $(f \comp_a g)\reg{a}{b} = f\reg{a}{b} \comp_b g\reg{a}{b},$ and
        \item\label{itm:comp-monoidal} monoidal with respect to the lax structure $\otimes_{a,b;c}$, i.e. $(f \comp_a g) \otimes_{a,b;c} (h \comp_b i) = (f \otimes_{a,b;c} h) \comp_c (g \otimes_{a,b;c} i),$ whenever $a \oplus b \leqslant c$.
      \end{enumerate}
  \item A monoidal natural transformation $\mathsf{id} : L \Rightarrow \mathbb{C}$, which is
      \begin{enumerate}[label=(\roman*), ref=\theenumi.(\roman*)]
        \item\label{itm:id-map} an element $(\id_X)_a \in \mathbb{C}_a(X;X)$, for each grade $a \in \mathbf{E}$ and object $X$,
        \item\label{itm:id-nat} natural in the grade, i.e. $((\id_X)_a)\reg{a}{b} = (\id_X)_b$ whenever $a \leqslant b$,
        \item\label{itm:id-monoidal} monoidal with respect to $\otimes_{a,b;c}$, i.e. $(\id_X)_a \otimes_{a,b;c} (\id_Y)_b = (\id_{X \otimes Y})_c$ whenever $a \oplus b \leqslant c$,
        \item\label{itm:id-lax} and compatible with the lax monoidal unit above, i.e. $\eta_c = (\id_I)_c$, for all $c \in \mathbf{E}$.
      \end{enumerate}
  \item\label{itm:mon-laws} the monoid laws for $\comp$ and $\id$, namely
      \begin{enumerate}[label=(\roman*), ref=\theenumi.(\roman*)]
        \item\label{itm:comp-assoc} associativity in each grade: for composable $f,g,h \in \mathbb{C}_a$, we have $(f \comp_a g)\comp_a h = f \comp_a (g \comp_a h)$,
        \item\label{itm:id-units} left and right unit in each grade: $(\id_X)_a \comp_a f = f$ and $f \comp_a (\id_Y)_a = f$ for all $f \in \mathbb{C}_a(X;Y)$.
      \end{enumerate}
  \end{enumerate}
    This data defines an $E$-graded monoidal category:
  by \cref{itm:lax-functor-homset,itm:comp-map,itm:id-map,itm:mon-laws}, for each $a \in E$ we have a category $ℂ_a$ with objects $\mathbb{C}\obj$;
  \cref{itm:lax-functor-regrade,itm:comp-nat,itm:id-nat} give identity-on-objects regrading functors $(-)\reg{a}{b} : ℂ_a \to ℂ_b$;
  \cref{itm:lax-tensor} gives  partial monoidal products with $c = a \oplus b $, i.e. $\otimes_{a,b} := \otimes_{a,b;a \oplus b}$.
  For the laws, \textsc{Reg-Act} is \cref{itm:lax-functor-regrade-laws}.
  For \textsc{Reg-$\otimes$} we have
  $$(f \otimes g)\reg{a \oplus b}{c \oplus d}
  := (f \otimes_{a,b;a \oplus b} g)\reg{a \oplus b}{c \oplus d}
  \overset{\ref{itm:tensor-nat}}{=} f \otimes_{a,b;c \oplus d} g
  \overset{\ref{itm:equiv-compat}}{=} (f\reg{a}{c}) \otimes_{c,d;c \oplus d} (g\reg{b}{d})
  =: f\reg{a}{c} \otimes g\reg{b}{d}.$$
\textsc{$\otimes$-U-A} follows from \cref{itm:id-lax,itm:lax-laws-unit,itm:lax-laws-assoc};
\textsc{$\otimes$-Id} follows from \cref{itm:id-monoidal};
\textsc{Inter} follows from \cref{itm:comp-monoidal} with $c = a \oplus b$.
For the converse, the idea is that given an $E$-graded monoidal category, we can define the monoidal products required above via regrading, $f \otimes_{a,b;c} g := (f \otimes_{a,b} g)\reg{a \oplus b}{c}$.
The lax monoidal unit $\eta_c$ (\cref{itm:lax-unit}) is determined by the sequential unit via \cref{itm:id-lax}, and \cref{itm:unit-nat} is then subsumed by \cref{itm:id-nat} at $X = I$.
We verify that the necessary laws hold in \Cref{app:egradismon}. These two processes are mutually inverse: starting from an $E$-graded monoidal category, we have $f \otimes_{a,b;a \oplus b} g = (f \otimes g)\reg{a \oplus b}{a \oplus b} = f \otimes g$ by \textsc{Reg-Act}; conversely, starting from a lax monoidal functor with laxator $\otimes_{a,b;c}$, the recovered laxator is $(f \otimes_{a,b;a \oplus b} g)\reg{a \oplus b}{c} = f \otimes_{a,b;c} g$ by \Cref{itm:tensor-nat}.
\end{proof}

\begin{remark}
  \emph{Locally indexed categories} \cite{levy2019locally} axiomatize a notion of category in which hom-sets $\mathbb{C}_v(A;B)$ carry a grade $v$, typically an object of a monoidal category $\mathcal{V}$, equipped with regrading morphisms and sequential composition operations that are homogeneous in the grade: $\mathbb{C}_v(A;B) \times \mathbb{C}_v(B;C) \to \mathbb{C}_v(A;C)$. Formally, a locally $\mathcal{V}$-indexed category is a category enriched in the category of presheaves on $\mathcal{V}$, equipped with its pointwise cartesian product ($\times$).  
  \emph{Locally graded categories} \cite{wood1976indicial,levy2019locally,10.1007/978-3-031-16912-0_4} are instead enriched in the convolution monoidal product on presheaves ($*$), which allows grades to \emph{combine} under sequential composition. As noted above, these two monoidal structures form a duoidal structure \cite{street12:linking}. Whenever one has a duoidal category $(\mathcal{V}, \ast, J, \times, I)$, the category $(\mathcal{V},*,J)$-$\cat{Cat}$ has a monoidal structure given by $(\times, I)$, as in Batanin and Markl \cite[Section 3]{BATANIN20121811}. In the case of $\mathcal{V} = [\mathbf{E},\Set]$, monoids in this monoidal category are a ``flipped'' notion of graded monoidal category in which grades combine sequentially but not monoidally.
\end{remark}

\section{Conclusion and future work}
We have introduced monoidal categories graded by partial commutative monoids: a uniform framework for monoidal categories, effectful categories, and Freyd categories. By varying the grading PCM, we have modelled non-interfering parallelism, bounded resource usage, and other resource-sensitive settings, going beyond the two-element grading of effectful categories.

\Cref{lem:global-category-structure} raises the possibility of equipping an "$E$-graded monoidal category" with sequential composition operations of the type $ℂ_a(X;Y) \times ℂ_b(Y;Z) \to ℂ_{a \vee b}(X;Z),$ suggesting a notion of monoidal categories graded in algebraic structures equipped with two (partial) monoid structures, with one operation grading the monoidal composition, and the other sequential composition. A similar notion in the total case has appeared in a preprint of the third author with Di Lavore \cite{dilavore2025timingpinwheeldoublecategories}. \Cref{sec:gradprom} suggests a concrete approach to defining such ``partial duoid'' graded monoidal categories, by taking an $\mathbf{E}$ to be an appropriate (thin) \emph{produoidal} category, which again induces a duoidal structure on presheaves \cite{booker2013tannaka,produoidal}. This might provide a common roof for our notion and that of Sarkis and Zanasi \cite{sarkis_et_al:LIPIcs.CALCO.2025.5} in which grades combine totally, by the same operation, under both sequential and monoidal composition.

Most of our examples of $E$-graded monoidal categories have been agnostic towards the definition of objects and morphisms in the category. It would be nice to give conditions under which effectful categories can be refined to non-trivially $E$-graded monoidal categories. The work of Breuvart, McDermott and Uustalu on canonical gradings of monads \cite{Breuvart_2023} may be relevant in this regard.

In previous work \cite{earnshaw2025resourceful}, we constructed free effectful categories over what might be termed $\mathscr{P}(X)$-graded \emph{signatures} (or ``effectful signatures''). In particular, we showed that every effectful category has a non-trivial underlying $\mathscr{P}(X)$-graded signature, where the set $X$ is given by taking maximal cliques in a graph determined by the morphisms of an effectful category. Thus we may wonder if there is a further right adjoint to the functor $R$ of \Cref{coreflective}, defined using this construction. However, it seems that the morphisms of $\GradMonTop$ are too strong for this to be the case, since their assignment on morphisms is governed by a morphism of \pcms{} whereas the morphisms of signatures in our paper \cite{earnshaw2025resourceful} were necessarily weaker, only preserving orthogonality. It remains to be seen how to relate these two notions of morphism.

Partial commutative monoids are used extensively in separation logic to model resources \cite{10.5555/645683.664578,10.1145/2775051.2676980}. Establishing substantive connections to separation logic is also a promising direction for future work. In particular, a natural next step would be to understand the relation between separating conjunction and partially defined monoidal products arising from grading by appropriate separation algebras.

\bibliographystyle{./entics}
\bibliography{main}

\appendix

\section{Monoidal, premonoidal, and effectful categories}
\subsection{Monoidal categories and their morphisms}
A strict monoidal category is usually defined to be a monoidal category for which associators and unitors are identities. The following elementary reformulation (see e.g. Román \cite{marioThesis}) is more convenient for our purposes.

\begin{definition} \label{moncat}
    A ""strict monoidal category"" $ℂ$ consists of a monoid of objects, $(ℂ\obj, ⊗, I)$, a collection of morphisms $ℂ(X; Y)$, for every pair of objects $X, Y ∈ ℂ\obj$, (families of) operations for the sequential and parallel composition of morphisms, respectively
  \begin{gather*}
    (\comp) : ℂ(X; Y) × ℂ(Y ; Z) → ℂ(X; Z),\ \text{and} \\
    (\otimes) : ℂ(X; Y) × ℂ(X'; Y') → ℂ(X ⊗ X'; Y ⊗ Y'),
  \end{gather*}
  and a family of identity morphisms, $\id_{X} ∈ ℂ(X; X)$. This data must satisfy the following axioms
  \begin{itemize}
    \item sequential composition is unital, $f ⨾ \id_{Y} = f$ and $\id_{X} ⨾ f = f$,
    \item sequential composition is associative, $f ⨾ (g ⨾ h) = (f ⨾ g) ⨾ h$,
    \item parallel composition is unital, $f ⊗ \id_{I} = f$ and $\id_{I} ⊗ f = f$,
    \item parallel composition is associative, $f ⊗ (g ⊗ h) = (f ⊗ g) ⊗ h$,
    \item parallel composition and identities interchange, $\id_{A} ⊗ \id_{B} = \id_{A⊗B}$, and
    \item parallel composition and sequential composition interchange, \\
    $(f ⨾ g) ⊗ (f' ⨾ g') = (f ⊗ f') ⨾ (g ⊗ g').$
  \end{itemize}
  Parallel composition is also called the ""tensor product"" or ""monoidal product"".
\end{definition}

\begin{definition}
  A ""strict monoidal functor"" $F : \mathbb{C} \to \mathbb{D}$ between strict monoidal categories is a functor that is a monoid homomorphism on objects, i.e. $F(A \otimes B) = F(A) \otimes F(B)$ and $F(I) = I$, and preserves the monoidal product on morphisms, i.e. $F(f \otimes g) = F(f) \otimes F(g)$.
\end{definition}

\begin{definition} \label{symstrmoncat}
  A ""symmetric strict monoidal category"" is a "strict monoidal category" further equipped with braiding maps $\sigma_{X,Y} \in ℂ(X \otimes Y ; Y \otimes X)$ for each pair of objects $X,Y$,
  which are natural and satisfy the axioms (i) $\sigma_{X,Y} \comp \sigma_{Y,X} = \id_{X \otimes Y}$, (ii) $(\sigma_{X,Y} \otimes \id_Z) \comp (\id_Y \otimes \sigma_{X,Z}) =  \sigma_{X,Y \otimes Z}$, (iii) $\sigma_{X,I} = \id_X$.
\end{definition}

\begin{definition}
  A ""symmetric strict monoidal functor"" is a "strict monoidal functor" which moreover preserves the braiding maps, $F(\sigma_{X,Y}) = \sigma_{FX,FY}$.
\end{definition}

\begin{definition}
  A ""cartesian monoidal category"" is a monoidal category whose monoidal product $A \otimes B$ is given by a chosen category theoretic product of $A$ and $B$, and in which the unit $I$ is a terminal object.
  A ""cartesian monoidal functor"" between cartesian monoidal categories is simply a "strict monoidal functor".
\end{definition}

\subsection{Premonoidal categories and their morphisms} 
For more details, see Power and Robinson \cite{power_robinson_1997}, or Román \cite{roman2022promonads}.

\begin{definition} \label{defn:premonoidal}
  A ""strict premonoidal category"" is a category $\mathbb{C}$ equipped with:
  \begin{itemize}
  \item for each pair of objects $A,B \in \mathbb{C}$ an object $A⊗B$,
  \item for each object $A \in \mathbb{C}$ a functor $A \ltimes -$ (``left-whiskering with $A$'') whose action on objects sends $B$ to $A ⊗ B$,
  \item for each object $A \in \mathbb{C}$ a functor $- \rtimes A$ (``right-whiskering with $A$'') whose action on objects sends $B$ to $B ⊗ A$, and
  \item a unit object $I$,
  \end{itemize}
  such that, for all morphisms $f$ and objects $A,B,C$ the following equations hold:
  \begin{itemize}
  \item $(A \ltimes f) \rtimes B = A \ltimes (f \rtimes B)$,
  \item $(A \otimes B) \ltimes f = A \ltimes (B \ltimes f)$,
  \item $f \rtimes (A \otimes B) = (f \rtimes A) \rtimes B$,
  \item $I \ltimes f = f = f \rtimes I$,
  \item $I ⊗ A = A = A ⊗ I$, and
  \item $A ⊗ (B ⊗ C) = (A ⊗ B) ⊗ C$.
  \end{itemize}
\end{definition}

\begin{definition} \label{defn:premonoidal_central}
  A morphism $f : A \to B$ in a "premonoidal category" is ""central"" if and only if for every morphism $g : C \to D$,
  \[\begin{aligned}
      (A \ltimes g){\comp}(f \rtimes D) &= (f \rtimes C){\comp}(B \ltimes g), ~\text{and}\\
      (g \rtimes A){\comp}(D \ltimes f) &= (C \ltimes f){\comp}(g \rtimes B).
      \end{aligned}\]
\end{definition}

\begin{definition} \label{defn:premonoidal_centre}
  The ""center of a premonoidal category"" is the wide subcategory of "central morphisms".
\end{definition}

\begin{definition}
  A ""symmetric strict premonoidal category"" $ℂ$ is a "strict premonoidal category" further equipped with central braiding maps $\sigma_{X, Y} \in ℂ(X \otimes Y; Y \otimes X)$ which are natural and satisfy the axioms (i) $\sigma_{X,Y} \comp \sigma_{Y,X} = \id_{X \otimes Y}$, (ii) $(\sigma_{X,Y} \rtimes Z) \comp (Y \ltimes \sigma_{X,Z}) =  \sigma_{X,Y \otimes Z}$, (iii) $\sigma_{X,I} = \id_X$.
\end{definition}

\begin{definition} \label{defn:premonoidalfunctor}
  A ""strict premonoidal functor"" $F : \mathbb{X} \to \mathbb{Y}$ between (strict) "premonoidal categories" is a functor that is a monoid homomorphism on objects, i.e. $F(A \otimes B) = F(A) \otimes F(B)$  and $F(I) = I$, and preserves whiskerings, i.e. $F(A \ltimes f) = F(A) \ltimes F(f)$ and $F(f \rtimes A) = F(f) \rtimes F(A)$.
\end{definition}

\begin{definition} \label{defn:sympremonoidalfunctor}
  A ""symmetric strict premonoidal functor"" is a "strict premonoidal functor" which moreover preserves the central braiding maps, $F(\sigma_{X,Y}) = \sigma_{FX,FY}.$
\end{definition}

\subsection{Effectful categories and their morphisms} \label{appeff}

\begin{definition}
  A strict ""effectful category"" $(\mathbb{V}, \mathbb{C}, \eta)$ is given by a
  "strict monoidal category" $\mathbb{V}$,
  \kl{strict premonoidal category} $\mathbb{C}$, and an
  identity on objects, "strict premonoidal functor" $\eta : \mathbb{V} \to
  \mathbb{C}$, such that the image of $\eta$ lands in the "center" of
  $\mathbb{C}$.
\end{definition}

\begin{definition}
  A ""symmetric strict effectful category"" is an effectful category $(\mathbb{V}, \mathbb{C}, \eta)$ in which $\mathbb{V}$ is equipped with a "symmetric strict monoidal structure", and $\mathbb{C}$ a "symmetric strict premonoidal structure", and $\eta : \mathbb{V} \to \mathbb{C}$ preserves the braidings.
\end{definition}

\begin{definition}
  A strict ""effectful functor"" $(F_0, F_1) : (\mathbb{V}, \mathbb{C}, \eta) \to (\mathbb{V}', \mathbb{C}', \eta')$ between "effectful categories" is a "strict monoidal functor" $F_0 : \mathbb{V} \to \mathbb{V}'$ and a "strict premonoidal functor" $F_1: \mathbb{C} \to \mathbb{C}'$ such that $\eta \comp F_1 = F_0 \comp \eta'$. In particular, we must have $F_0(X) = F_1(X)$ for all $X$ in $\mathbb{V}\obj = \mathbb{C}\obj$.
\end{definition}

\begin{definition}
  A ""symmetric strict effectful functor"" is a strict "effectful functor" $(F_0, F_1)$ for which $F_0$ is a "symmetric strict monoidal functor" and $F_1$ is a "symmetric strict premonoidal functor".
\end{definition}

Write $""\SymEff""$ for the category of "symmetric strict effectful categories" and "symmetric strict effectful functors".

\begin{definition}
  A (strict) ""Freyd category"" $(\mathbb{V}, \mathbb{C}, \eta)$
  is a "symmetric strict effectful category" in which $\mathbb{V}$ is a cartesian monoidal category.
\end{definition}

\begin{definition}
  A ""functor of Freyd categories"" $(F_0, F_1) : (\mathbb{V}, \mathbb{C}, \eta) \to (\mathbb{V}', \mathbb{C}', \eta')$
  is a "symmetric strict effectful functor" for which $F_0$ is a "cartesian monoidal functor".
\end{definition}

Write $""\Freyd""$ for category of "Freyd categories" and functors between them.

 \section{Details for \Cref{sec:effectful}}

\subsection{Proof of \Cref{lem:interchange}} \label{app:lem:interchange}
\begin{proof}
      \begin{align*}
  &(f\reg{0}{a} \otimes \id_{X'}) ⨾ (\id_Y \otimes g) \\
  &= (f\reg{0}{a} \otimes \id_{X'}\reg{0}{0}) ⨾ (\id_Y \otimes g)
  &&\text{(\textsc{Reg-Act})} \\
  &= (f \otimes \id_{X'})\reg{0 \oplus 0 }{a \oplus 0} ⨾ (\id_Y \otimes g)
  &&\text{(\textsc{Reg-$\otimes$})} \\
  &= (f \otimes \id_{X'})\reg{0 \oplus 0}{0 \oplus a} ⨾ (\id_Y \otimes g)
  &&\text{(PCM comm.)} \\
  &= (f\reg{0}{0} \otimes \id_{X'}\reg{0}{a}) ⨾ (\id_Y \otimes g)
    &&\text{(\textsc{Reg-$\otimes$})} \\
  &= (f ⨾ \id_Y) \otimes (\id_{X'}\reg{0}{a} ⨾ g)
    &&\text{(\textsc{Reg-Act}, \textsc{Inter})} \\
  &= f \otimes g,
\end{align*}
and analogous reasoning starting from $(\id_X \otimes g) \comp (f\reg{0}{a} \otimes \id_{Y'})$ similarly arrives at $f \otimes g$.
\end{proof}

 \subsection{Proof of \Cref{symtop}} \label{app:symtop}
\begin{proof}
        \Cref{toppremon} gives that $\mathbb{C}_a$ is premonoidal. Write $\sigma_{X,Y} \in \mathbb{C}_0(X\otimes Y,Y \otimes X)$ for the braiding of $\mathbb{C}$. Then the braidings of $\mathbb{C}_a$ are given by $(\sigma_{X,Y})_0^a \in \mathbb{C}_a(X \otimes Y,Y \otimes X)$. These satisfy the required axioms of braidings, using axioms for regrading and the axioms of braidings in $ℂ_0$,
  \begin{align*}
    &(i) \quad (\sigma_{X,Y})_0^a \comp (\sigma_{Y,X})_0^a
      = (\sigma_{X,Y}\comp\sigma_{Y,X})_0^a
      = (\id_{X \otimes Y})_0^a
      = \id_{X \otimes Y},\\
    & (ii) \quad ((\sigma_{X,Y})_0^a \rtimes Z) \comp (Y \ltimes (\sigma_{X,Z})_0^a) :=
      ((\sigma_{X,Y})_0^a \otimes_{a,0} \id_Z) \comp (\id_Y \otimes_{0,a} (\sigma_{X,Z})_0^a) \\
      &= (\sigma_{X,Y} \otimes \id_Z)_0^a \comp (\id_{Y} \otimes \sigma_{X,Z})_0^a
      = ((\sigma_{X,Y} \otimes \id_Z) \comp (\id_{Y} \otimes \sigma_{X,Z}))_0^a
    = (\sigma_{X,Y\otimes Z})_0^a  \\
      & (iii) \quad  (\sigma_{X,I})_0^a = (\id_{X})_0^a
      \end{align*}
      and for any $f \in \mathbb{C}_a(X;Y)$ we have naturality, via the symmetry axiom for symmetric $E$-graded monoidal categories,
      $$(f \rtimes Z)\comp(\sigma_{Y,Z})_0^a
      =  (f \otimes \id_Z)\comp(\sigma_{Y,Z})_0^a
      = (\sigma_{X,Z})_0^a\comp(\id_Z \otimes f)
      = (\sigma_{X,Z})_0^a\comp(Z \ltimes f).$$

      It remains to show that these braidings are "central" in $\mathbb{C}_a$. Let $g \in \mathbb{C}_a(A;B)$. Unfolding the whiskerings in $\mathbb{C}_a$ and applying \Cref{lem:interchange} we obtain
      \begin{align*}
	    ((X \otimes Y) \ltimes g)\comp((\sigma_{X,Y})_0^a \rtimes B)
	    &:= (\id_{X \otimes Y} \otimes g)\comp((\sigma_{X,Y})_0^a \otimes \id_B) \\
	    &= ((\sigma_{X,Y})_0^a \otimes \id_A)\comp(\id_{Y \otimes X} \otimes g) \\
	    &=: ((\sigma_{X,Y})_0^a \rtimes A)\comp((Y \otimes X) \ltimes g),
      \end{align*}
and similarly for the other interchange equation.
Hence $(\sigma_{X,Y})_0^a$ is central, so $\mathbb{C}_a$ is a "symmetric premonoidal category".
\end{proof}

    \section{Details for \Cref{sec:globalcat}}

  \subsection{Proof of \Cref{directedc}} \label{app:directedc}
  \begin{proof}
  Write $[\langle a , f \rangle] \in \overline{ℂ}(X;Y)$ for the equivalence class of a morphism $f \in ℂ_a(X;Y)$. Let $[\langle a, f \rangle] \in \overline{ℂ}(X;Y)$ and $[\langle b, g \rangle] \in \overline{ℂ}(Y;Z)$. We define composition as follows.
  Since $E$ is directed, we can choose a $c$ such that $a,b \leqslant c$.
  Define $[\langle a, f \rangle] \comp [\langle b, g \rangle] := [\langle c, f\reg{a}{c} \comp_c g\reg{b}{c} \rangle]$. We must show this is well defined. Firstly, if $c'$ is another $a,b \leqslant c'$, then by directedness we can again choose $c, c' \leqslant d$, and by the axioms of $E$-graded monoidal categories we have $$(f\reg{a}{c} \comp_c g\reg{b}{c})\reg{c}{d} = (f\reg{a}{c})\reg{c}{d} \comp_d (g\reg{b}{c})\reg{c}{d} = f\reg{a}{d} \comp_d g\reg{b}{d} = (f\reg{a}{c'})\reg{c'}{d} \comp_d (g\reg{b}{c'})\reg{c'}{d} = (f\reg{a}{c'} \comp_{c'} g\reg{b}{c'})\reg{c'}{d},$$ so $[\langle c, f\reg{a}{c} \comp g\reg{b}{c} \rangle] = [\langle d, f\reg{a}{d} \comp g\reg{b}{d}\rangle] = [\langle c', f\reg{a}{c'} \comp g\reg{b}{c'} \rangle]$.
  
  Now let $\langle p,h \rangle$ and $\langle q, i \rangle$ be different representatives of the equivalence classes $[\langle a, f \rangle]$ and $[\langle b, g \rangle]$ respectively. For $p,q \leqslant r$, we must show $[\langle r, h\reg{p}{r} \comp i\reg{q}{r} \rangle] = [\langle c, f\reg{a}{c} \comp g\reg{b}{c} \rangle]$. Again, by directedness choose $c,r \leqslant z$, then we shall show that both equivalence classes are equal to $[\langle z, f\reg{a}{z} \comp g\reg{b}{z}\rangle]$. On the one hand, from $c \leqslant z$ we have $[\langle c, f\reg{a}{c} \comp g\reg{b}{c} \rangle] = [\langle z, f\reg{a}{z} \comp g\reg{b}{z}\rangle]$, establishing one of the desired equalities. From $r \leqslant z$ and the axioms for regrading, we have $[\langle r, h\reg{p}{r} \comp i\reg{q}{r} \rangle] = [\langle z, h\reg{p}{z} \comp i\reg{q}{z} \rangle]$, so it remains to show $[\langle z, h\reg{p}{z} \comp i\reg{q}{z} \rangle] = [\langle z, f\reg{a}{z} \comp g\reg{b}{z}\rangle]$.

  We have assumed $[\langle a, f\rangle] = [\langle p, h \rangle]$, so there must exist a sequence $\langle a_0, f_0 \rangle \equiv ... \equiv \langle a_n, f_n \rangle$, with $a_0 = a, f_0 = f$ and $a_n = p, f_n = h$, and either $a_{i-1} \leqslant a_i$ and $f_i = (f_{i-1})\reg{a_{i-1}}{a_i}$, or $a_i \leqslant a_{i-1}$ and $f_{i-1} = (f_i)\reg{a_{i}}{a_{i-1}}$, and similarly for $g$ (with sequence of grades $b = b_0, ..., b_m = i$, say). By directedness of $E$, there exists an element $w$ such that $w \geqslant z, a_i, b_k$ for all $0 \leqslant i \leqslant n$ and $0 \leqslant k \leqslant m$. Observe that for any step in the chain, say $\langle a_{i-1}, f_{i-1} \rangle \equiv \langle a_i, f_i \rangle$, the morphisms become equal when regraded to $w$. Explicitly, if $a_{i-1} \leqslant a_i$ and $f_i = (f_{i-1})\reg{a_{i-1}}{a_i}$, then by the axioms of regrading
  $$(f_i)\reg{a_i}{w} = ((f_{i-1})\reg{a_{i-1}}{a_i})\reg{a_i}{w} = (f_{i-1})\reg{a_{i-1}}{w}.$$
  
The case $a_i \leqslant a_{i-1}$ follows symmetrically. Applying this to each pair we obtain, $f\reg{a}{w} = f_0\reg{a_0}{w} = ... = f_n\reg{a_n}{w} = h\reg{p}{w}$.
Similar reasoning for $g$ establishes $g\reg{b}{w} = i\reg{q}{w}$, so finally we have,
    $$[\langle z, f\reg{a}{z} \comp g\reg{b}{z} \rangle] = [\langle w, (f\reg{a}{z} \comp g\reg{b}{z})\reg{z}{w} \rangle] = [\langle w, f\reg{a}{w} \comp g\reg{b}{w} \rangle] = [\langle w, h\reg{p}{w} \comp i\reg{q}{w} \rangle] = [\langle w, (h\reg{p}{z} \comp i\reg{q}{z})\reg{z}{w} \rangle] = [\langle z, h\reg{p}{z} \comp i\reg{q}{z} \rangle],$$
  establishing that composition is well defined.
  For associativity, let $[\langle a,f \rangle]$, $[\langle b,g \rangle]$, $[\langle c,h \rangle]$ be composable. We show both orders of composition are equal to $[\langle t, f_a^t \comp_t g_b^t \comp_t h_c^t \rangle],$ for some $t \geqslant a,b,c$, recalling that $\comp_t$ is associative axiomatically.
  For instance, compute $( [\langle a,f \rangle] \comp [\langle b, g \rangle] ) \comp [\langle c,h \rangle]$ by taking $d \geqslant a,b$ and then $t \geqslant d,c$, giving $[\langle t, (f\reg{a}{d} \comp g\reg{b}{d})\reg{d}{t} \comp h\reg{c}{t} \rangle] = [\langle t, f\reg{a}{t} \comp g\reg{b}{t} \comp h\reg{c}{t} \rangle]$, and similarly for the other parenthesisation. Therefore composition is associative.

  The identity at $X$ is $[\langle 0,\id_X \rangle]$. Given $[\langle a,f \rangle] : X \to Y$, and taking $a \geqslant 0,a$ we have left unitality,
  $ [\langle 0, \id_X \rangle] \comp [ \langle a,f \rangle]
    = [\langle a, (\id_X)\reg{0}{a} \comp_a f\reg{a}{a} \rangle]
    = [\langle a,f \rangle],$
    with right unitality following similarly.
  \end{proof}

  \subsection{Proof of \Cref{globmon}} \label{app:globmon}
  \begin{proof}
    A total commutative monoid has a directed "extension order" since $a,b \leqslant a \oplus b$, so $\overline{ℂ}$ exists by \Cref{directedc}. Define
  $[\langle a, f \rangle] \otimes [\langle b, g \rangle] := [\langle a \oplus b, f \otimes_{a,b} g \rangle].$
  To show that this is well defined, it suffices to check compatibility with the generating relation. If $a \leqslant c$ and $\langle a,f \rangle \equiv \langle c, f\reg{a}{c} \rangle$, then
  $$(f \otimes g)\reg{a \oplus b}{c \oplus b} =
    f\reg{a}{c} \otimes g\reg{b}{b} =
    f\reg{a}{c} \otimes g$$
  by \textsc{Reg-$\otimes$} and \textsc{Reg-Act}, so
    $$[\langle a \oplus b, f \otimes g \rangle] =
    [\langle c \oplus b, f\reg{a}{c} \otimes g \rangle].$$
  The argument in the second variable is identical, and therefore the monoidal product descends to equivalence classes.

  The monoidal unit is $I$, whose identity morphism in $\overline{ℂ}$ is $[\langle 0,\id_I \rangle]$. Associativity on morphisms follows immediately from $\otimes$-\textsc{U-A}. For instance,
    $$([\langle a,f \rangle] \otimes [\langle b,g \rangle]) \otimes [\langle c,h \rangle] =
    [\langle (a \oplus b) \oplus c, (f \otimes g) \otimes h \rangle] =
    [\langle a \oplus (b \oplus c), f \otimes (g \otimes h) \rangle] =
    [\langle a,f \rangle] \otimes ([\langle b,g \rangle] \otimes [\langle c,h \rangle]).$$
  Unitality for the monoidal product is likewise inherited from $\otimes$-\textsc{U-A}:
    $$[\langle a,f \rangle] \otimes [\langle 0,\id_I \rangle] = [\langle a \oplus 0, f \otimes \id_I \rangle] = [\langle a,f \rangle],$$
  and similarly on the left, and we have that
  $[\langle 0,\id_X \rangle] \otimes [\langle 0,\id_Y \rangle] = [\langle 0, \id_X \otimes \id_Y \rangle] = [\langle 0, \id_{X \otimes Y} \rangle]$
  by $\otimes$-\textsc{ID}.

  It remains to check interchange with composition in $\overline{ℂ}$. Let $[\langle a,f \rangle]$, $[\langle a',h \rangle]$, $[\langle b,g \rangle]$, and $[\langle b',k \rangle]$ be composable. Using the definition of composition from \Cref{directedc}, choose the upper bound $a \oplus a'$ of $a,a'$ and $b \oplus b'$ of $b,b'$. Then
  \begin{align*}
    &([\langle a,f \rangle] \comp [\langle a',h \rangle]) \otimes ([\langle b,g \rangle] \comp [\langle b',k \rangle]) \\
    &= [\langle a \oplus a', f\reg{a}{a \oplus a'} \comp h\reg{a'}{a \oplus a'} \rangle]
      \otimes
      [\langle b \oplus b', g\reg{b}{b \oplus b'} \comp k\reg{b'}{b \oplus b'} \rangle] \\
    &= [\langle a \oplus a' \oplus b \oplus b', (f\reg{a}{a \oplus a'} \comp h\reg{a'}{a \oplus a'}) \otimes (g\reg{b}{b \oplus b'} \comp k\reg{b'}{b \oplus b'}) \rangle] \\
    &= [\langle a \oplus b \oplus a' \oplus b', (f \otimes g)\reg{a \oplus b}{(a \oplus a') \oplus (b \oplus b')} \comp (h \otimes k)\reg{a' \oplus b'}{(a \oplus a') \oplus (b \oplus b')} \rangle],
  \end{align*}
  by \textsc{Reg-$\otimes$}, \textsc{Inter}, associativity and commutativity of $\oplus$. This is exactly the composite
    $$[\langle a \oplus b, f \otimes g \rangle] \comp [\langle a' \oplus b', h \otimes k \rangle] =
    ([\langle a,f \rangle] \otimes [\langle b,g \rangle]) \comp ([\langle a',h \rangle] \otimes [\langle b',k \rangle])$$
  using the common upper bound $a \oplus a' \oplus b \oplus b'$.
  Therefore $\overline{ℂ}$ is a strict monoidal category.
  \end{proof}

\section{Details for \Cref{sec:gradprom}} \label{app:gradpromproof}

\subsection{Proof of \Cref{pcmprom}} \label{app:pcmprom}
\begin{proof}
  As regards functoriality, we have
  \begin{itemize}
  \item if $a' \leqslant a$ and $P(a,b;c) = \top$, then $a' \oplus b \leqslant a \oplus b \leqslant c$ by \Cref{lem:leqslant-cong}, hence $P(a',b;c) = \top$, and similarly for $b$,
  \item if $c \leqslant c'$ and $P(a,b;c) = \top$, then $a \oplus b \leqslant c$ implies $a \oplus b \leqslant c'$, so $P(a,b;c') = \top$,
  \item if $c\leqslant c'$ and $I(c) = \top$ then $0 \leqslant c \leqslant c'$, so $I(c') = \top$.
  \end{itemize}

  For associativity, we need to show, for arbitrary elements $a,b,c,d \in \mathbf{E}$ $$\exists\ {x \in \mathbf{E}},\ a \oplus x \leqslant d\ \wedge\ b \oplus c \leqslant x \iff \exists\ {x \in \mathbf{E}},\ a \oplus b \leqslant x\ \wedge\ x \oplus c \leqslant d.$$

  Assume the left hand side holds with witness $x$, then by \Cref{lem:leqslant-cong} we have $(a \oplus b) \perp c$ and $a \oplus b \oplus c \leqslant a \oplus x \leqslant d$, and so $a \oplus b$ witnesses the right hand side. If the right hand side holds with witness $x$, then \Cref{lem:leqslant-cong} similarly gives that $b \oplus c$ witnesses the left hand side.
  
  For right unitality, since $I$ is constantly true, we need to show $$\exists\ {x \in \mathbf{E}},\ a \oplus x \leqslant b \iff a \leqslant b.$$ Let the left hand side hold with witness $x$. Then since $0 \leqslant x$, we have $a = a \oplus 0 \leqslant a \oplus x \leqslant b$ from \Cref{lem:leqslant-cong}. Conversely, if the right hand side holds, $0$ witnesses the left hand side. Left unitality is analogous.
\end{proof}

\subsection{Proof of \Cref{egradismon} (cont.)} \label{app:egradismon}
\begin{proof}
  For the converse, let an $E$-graded monoidal category be given in the sense of \Cref{gradmoncat}. Then immediately we have sets $ℂ_a(X;Y)$, associative and unital compositions $(\comp)_a$ with identities $(\id_X)_a := (\id_X)_0^a$, and regrading functors $(-)\reg{a}{b}$. For every $a \oplus b \leqslant c$, define
  $$f \otimes_{a,b;c} g := (f \otimes_{a,b} g)\reg{a \oplus b}{c},
    \qquad
    \eta_c := (\id_I)_0^c.$$
  As before, \cref{itm:lax-functor-regrade-laws} is exactly \textsc{Reg-Act}.
  \Cref{itm:unit-nat} follows from an application of \textsc{Reg-Act}.
  For \cref{itm:equiv-compat}, let $a \leqslant a'$, $b \leqslant b'$, and $a' \oplus b' \leqslant c$.
  Then $a \oplus b \leqslant c$, so both sides below are defined, and
  \begin{align*}
    (f\reg{a}{a'}) \otimes_{a',b';c} (g\reg{b}{b'})
    &:= ((f\reg{a}{a'}) \otimes_{a',b'} (g\reg{b}{b'}))\reg{a' \oplus b'}{c} \\
    &= ((f \otimes_{a,b} g)\reg{a \oplus b}{a' \oplus b'})\reg{a' \oplus b'}{c}
    &&\text{(\textsc{Reg-$\otimes$})} \\
    &= (f \otimes_{a,b} g)\reg{a \oplus b}{c}
    &&\text{(\textsc{Reg-Act})} \\
    &=: f \otimes_{a,b;c} g.
  \end{align*}
  \Cref{itm:tensor-nat} follows similarly, by application of \textsc{Reg-Act}.
  For \cref{itm:lax-laws-assoc}, suppose grades satisfying the hypotheses are given. Then, we have
    \begin{align*}
      (f \otimes_{a,b;x} g) \otimes_{x,c;d} h
      &:= (( (f \otimes g)\reg{a \oplus b}{x} ) \otimes h)\reg{x \oplus c}{d} \\
      &:= (( (f \otimes g)\reg{a \oplus b}{x} ) \otimes h\reg{c}{c})\reg{x \oplus c}{d}
      && \textsc{(Reg-Act)} \\
      &= (((f \otimes g) \otimes h)\reg{(a \oplus b) \oplus c}{x \oplus c})\reg{x \oplus c}{d}
      &&\textsc{(Reg-$\otimes$)} \\
      &= ((f \otimes g) \otimes h)\reg{(a \oplus b) \oplus c}{d},
      && \textsc{(Reg-Act)}
    \end{align*}
    An analogous derivation from the right hand side yields the same result, after applying $\otimes$-\textsc{U-A} and associativity of \pcms{}.
    For \cref{itm:lax-laws-unit}, let $a \oplus b \leqslant c$. Then
    \begin{align*}
      f \otimes_{a,b;c} \eta_b
      &:= (f \otimes_{a,b} (\id_I)_b)\reg{a \oplus b}{c} \\
      &= (f \otimes_{a,b} (\id_I)\reg{0}{b})\reg{a \oplus b}{c}
      &&\text{(def. of $(\id_I)_b$)} \\
      &= ((f \otimes \id_I)\reg{a}{a \oplus b})\reg{a \oplus b}{c}
      &&\text{(\textsc{Reg-$\otimes$}, \textsc{Reg-Act})} \\
      &= (f \otimes \id_I)\reg{a}{c}
      &&\text{(\textsc{Reg-Act})} \\
      &= f\reg{a}{c}.
      &&\text{($\otimes$-\textsc{U-A})}
    \end{align*}
    and similarly for left unitality.
    For \cref{itm:id-monoidal}, let $a \oplus b \leqslant c$. Then
    \begin{align*}
      (\id_X)_a \otimes_{a,b;c} (\id_Y)_b
      &:= ((\id_X)_a \otimes_{a,b} (\id_Y)_b)\reg{a \oplus b}{c} \\
      &= ((\id_X)\reg{0}{a} \otimes_{a,b} (\id_Y)\reg{0}{b})\reg{a \oplus b}{c}
      &&\text{(def. of grade-$a,b$ identities)} \\
      &= ((\id_X \otimes \id_Y)\reg{0}{a \oplus b})\reg{a \oplus b}{c}
      &&\text{(\textsc{Reg-$\otimes$})} \\
      &= (\id_X \otimes \id_Y)\reg{0}{c}
      &&\text{(\textsc{Reg-Act})} \\
      &= (\id_{X \otimes Y})\reg{0}{c}
      &&\text{($\otimes$-\textsc{ID})} \\
      &=: (\id_{X \otimes Y})_c.
    \end{align*}
    \Cref{itm:id-lax} is immediate.
    For \cref{itm:comp-monoidal}, let $a \oplus b \leqslant c$. Then we have
  \begin{align*}
    (f \comp_a g) \otimes_{a,b;c} (h \comp_b i)
    &:= ((f \comp_a g) \otimes_{a,b} (h \comp_b i))\reg{a \oplus b}{c} \\
    &= ((f \otimes_{a,b} h) \comp_{a \oplus b} (g \otimes_{a,b} i))\reg{a \oplus b}{c}
      &&\text{(\textsc{Inter})}\\
    &= (f \otimes_{a,b} h)\reg{a \oplus b}{c} \comp_c (g \otimes_{a,b} i)\reg{a \oplus b}{c}
      &&\text{(reg. func.)}\\
    &=: (f \otimes_{a,b;c} h) \comp_c (g \otimes_{a,b;c} i).
  \end{align*}
    
  \Cref{itm:comp-nat} follows from the functoriality of regrading, \cref{itm:id-nat} follows from \textsc{Reg-Act}.
  We therefore have a monoid in $(\MonCatl(\mathbb{C}\obj\op \times \mathbb{C}\obj, ([\mathbf{E},\Set], \ast, J)), \circ, L)$.
\end{proof}

\end{document}